\documentclass[letterpaper, 10 pt, conference]{ieeeconf}  

\IEEEoverridecommandlockouts                              

\usepackage{cite}
\usepackage{amsmath,amssymb,amsfonts}
\usepackage{graphicx}
\usepackage{hyperref}
\hypersetup{hidelinks=true}
\usepackage{textcomp}

\usepackage{amsmath,amssymb,amsfonts}

\usepackage{graphicx}      
\usepackage{siunitx}
\usepackage{pgfplots}
\usepackage{tikz,amsmath}
\usepgfplotslibrary{fillbetween}
\usetikzlibrary{calc}
\usepackage{tikzscale}
\usepackage{algpseudocode}
\usepackage{algorithm}
\usepackage{subcaption} 
\usepackage{eso-pic}
\usepackage{mathdots}
\usepackage{accents}
\usepackage{array}
\usepackage{makecell}
\usepackage{comment}

\usetikzlibrary{intersections}
\usepackage{tkz-euclide}
\usetikzlibrary{shapes,arrows}
\usetikzlibrary{positioning}
\usetikzlibrary{spy}

\let\labelindent\relax
\usepackage[shortlabels]{enumitem}

\newtheorem{thm}{Theorem}
\newtheorem{lem}{Lemma}

\newtheorem{rem}{Remark}
\newtheorem{ass}{Assumption}

\pgfplotsset{compat=newest}
\pgfplotsset{plot coordinates/math parser=false}
\newlength\figureheight
\newlength\figurewidth

\newlength\defcolwidth

\newcommand*{\tran}{^{\mkern-1.5mu\mathsf{T}}}

\usepackage{tikz}
\usepackage{tikzscale}

\definecolor{myorange}{cmyk}{0,0.35,0.85,0} 
\definecolor{mypurple}{cmyk}{0.5,1,0,0} 

\definecolor{matblue1}{rgb}{0,0.4470,0.7410}
\definecolor{matred1}{rgb}{0.85,0.325,0.098}
\definecolor{matyel1}{rgb}{0.9290, 0.6940, 0.1250}
\definecolor{matpur1}{rgb}{0.4940, 0.1840, 0.5560}
\definecolor{matgre1}{rgb}{0.4660, 0.6740, 0.1880}
\definecolor{matblue2}{rgb}{0.3010, 0.7450, 0.9330}
\definecolor{matred2}{rgb}{0.6350, 0.0780, 0.1840}
\definecolor{matgrey1}{rgb}{0.5, 0.6, 0.7}
\definecolor{matpink1}{rgb}{1, 0.07, 0.65}
\definecolor{matblue3}{rgb}{0.07, 0.62, 1}

\newcommand{\yeldashdot}{\raisebox{2pt}{\tikz{\draw[-,matyel1,densely dash dot,line width = 0.9pt](0,0) -- (3mm,0);}}}
\newcommand{\purdashdot}{\raisebox{2pt}{\tikz{\draw[-,matpur1,densely dash dot,line width = 0.9pt](0,0) -- (3mm,0);}}}

\newcommand{\purdots}{\raisebox{2pt}{\tikz{\draw[-,matpur1,densely dotted,line width = 0.9pt](0,0) -- (3mm,0);}}}

\newcommand{\reddots}{\raisebox{2pt}{\tikz{\draw[-,matred1,densely dotted,line width = 0.9pt](0,0) -- (3mm,0);}}}

\newcommand{\reddash}{\raisebox{2pt}{\tikz{\draw[-,matred1,dashed,line width = 0.9pt](0,0) -- (3mm,0);}}}
\newcommand{\bluedash}{\raisebox{2pt}{\tikz{\draw[-,matblue1,dashed,line width = 0.9pt](0,0) -- (3mm,0);}}}

\newcommand{\yeldash}{\raisebox{2pt}{\tikz{\draw[-,matyel1,dashed,line width = 0.9pt](0,0) -- (3mm,0);}}}

\newcommand{\blackline}{\raisebox{2pt}{\tikz{\draw[-,black,solid,line width = 0.9pt](0,0) -- (3mm,0);}}}
\newcommand{\blueline}{\raisebox{2pt}{\tikz{\draw[-,matblue1,solid,line width = 0.9pt](0,0) -- (3mm,0);}}}
\newcommand{\redline}{\raisebox{2pt}{\tikz{\draw[-,matred1,solid,line width = 0.9pt](0,0) -- (3mm,0);}}}

\newcommand{\yelline}{\raisebox{2pt}{\tikz{\draw[-,matyel1,solid,line width = 0.9pt](0,0) -- (3mm,0);}}}
\newcommand{\purline}{\raisebox{2pt}{\tikz{\draw[-,matpur1,solid,line width = 0.9pt](0,0) -- (3mm,0);}}}
\newcommand{\greenline}{\raisebox{2pt}{\tikz{\draw[-,matgre1,solid,line width = 0.9pt](0,0) -- (3mm,0);}}}

\newcommand{\blackfulldot}{\raisebox{.7pt}{\tikz{\draw[-,black,fill=black,solid,line width = 1pt](0,0) circle (.5mm);}}}
\newcommand{\grefulldot}{\raisebox{.7pt}{\tikz{\draw[-,matgre1,fill=matgre1,solid,line width = 1pt](0,0) circle (.5mm);}}}

\graphicspath{{./Figures/}}

\DeclareMathOperator*{\argmax}{arg\,max}

\title{\LARGE \bf
Uncertainty-based perturb and observe \\ for data-driven optimization
}%

\author{Leontine Aarnoudse$^{1}$, Mark Haring$^{2}$, Nathan van de Wouw$^{3}$, and Alexey Pavlov$^{1}$
	\thanks{*This work was supported by the Research Council of Norway (RCN) through the project EXTREME EFFICIENCY: Data-driven optimization of industrial processes in time-varying environments (RCN project nr. 345272).}
	\thanks{$^{1}$Leontine Aarnoudse and Alexey Pavlov are with the Dept. for Geoscience, Norwegian University of Science and Technology, Trondheim, Norway. {\tt\small leontine.i.m.aarnoudse@ntnu.no}}%
	\thanks{$^{2}$Mark Haring is with the Dept. of Mathematics and Cybernetics, SINTEF Digital, Trondheim, Norway.
		}%
	\thanks{$^{3}$Nathan van de Wouw is with the Dept. of Mechanical Engineering, Eindhoven University of Technology, Eindhoven, The Netherlands.
		}%
}%

\begin{document}
	\AddToShipoutPictureBG*{%
	\AtPageUpperLeft{%
		\setlength\unitlength{1in}%
		\hspace*{\dimexpr0.5\paperwidth\relax}
		\makebox(0,-1)[c]{
			\parbox{\paperwidth}{ \centering
				Leontine Aarnoudse, Mark Haring, Nathan van de Wouw, and Alexey Pavlov, \\ Uncertainty-based perturb and observe for data-driven optimization, \\
				\textit{This work has been submitted to the IEEE for possible publication. Copyright may be transferred without notice, \\ after which this version may no longer be accessible.}}}%
}}

\maketitle%
\thispagestyle{empty}%
\pagestyle{empty}%

\setlength\defcolwidth{7.85cm}

\setlength\figurewidth{.9\defcolwidth}
\setlength\figureheight{.7\figurewidth}

\begin{abstract}
Data-based adaptive optimization methods hold great promise for the performance optimization of uncertain, time-varying processes. However, current methods are often based on continuous perturbation which is in general undesired for real-life (e.g., industrial) applications. In this paper, a new uncertainty-based perturb-and-observe method is developed that addresses this limitation and reduces the required number of perturbations, while retaining the capability to track time-varying optima. The method is based on the philosophy of `only perturbing when needed,' and is shown to converge to the optimum under mild conditions. A simulation-based case study on a photo-voltaic solar array demonstrates that it can outperform the standard perturb and observe approach as well as three other data-based optimization methods.
\end{abstract}

\section{Introduction}
\label{sec:intro}

Many industrial processes and systems are challenging to optimize due to their uncertain nature, which leads to performance limitations in terms of, for example, energy efficiency or throughput. Examples include solar arrays \cite{Killi2015}, melting ovens \cite{Yilmaz2023}, oil wells \cite{Krishnamoorthy2016} and dividing-wall distillation columns \cite{Halvorsen2025}. Often, the characteristics of these processes are time-varying because they depend on environmental conditions, feed compositions and upstream variations in the production line. In addition, in-depth understanding of the process as well as sensor measurements are often limited, in the sense that not all relevant process variables can be measured. This leads to large uncertainties when modeling these systems. As a result, robustness requirements limit the attainable performance with model-based approaches due to the conservatism needed to guarantee robustness to model uncertainties.

Model-free optimization methods can overcome some of the performance limitations associated with model-based optimization of uncertain processes. Common methods include extremum seeking control (ESC) \cite{Krstic2000,Scheinker2024} and perturb and observe (P\&O) \cite{Sternby1980,Hussein1995}, both of which actively probe the process to improve the operating conditions. These methods have been successfully applied to, a.o., wind energy systems \cite{Bafandeh2017,Mulders2019}, solar arrays \cite{Femia2005,Killi2015}, drilling \cite{Nystad2022}, ovens \cite{Yilmaz2023} and motion stages \cite{Hazeleger2022a}. 

The most common ESC method is the one discussed in, a.o., \cite{Krstic2000}, where the tunable process parameters are probed using a periodic dither signal, which is used to construct a gradient estimate or a direction of improvement of the performance cost. Typically, a form of averaging is used when obtaining a gradient estimate, such that the influence of measurement disturbances is small. This is often done using finite-order filters such as low-pass or high-pass filters \cite{Krstic2000,Tan2006}, but there exist extensions to, for example, moving average filters for nonlinear systems with a periodic steady-state \cite{Haring2013}. A limitation of ESC is that the required time-scale separation between the plant dynamics, gradient estimator and optimizer can lead to slow optimization, especially if the underlying transient dynamics of the process towards the steady-state are slow. There exist many methods to increase the convergence speed of ESC, for example by speeding up the learning dynamics \cite{Ghaffari2012}, online approximations of the system dynamics \cite{VanKeulen2020}, and replacing some experiments by kernel-based cost function approximations \cite{Weekers2025a}. 

P\&O is a more intuitive and straightforward alternative to ESC. P\&O, also known as hill-climbing, is often applied to problems such as maximum power-point tracking, see, e.g., \cite{Hussein1995,Femia2005,Killi2015}. The key idea is to simply steer the process parameters in one direction as long as this leads to improved performance, and change the direction when the performance degrades. One of the main advantages of P\&O compared to ESC is its intuitive nature, making it more accessible for industries that now often rely on human operators. In addition, P\&O is often faster, since it only requires one perturbation to determine a direction of improvement, and its sensitivity to measurement noise can be reduced by, e.g., averaging over multiple measurements \cite{Sternby1980}.

There are ample examples of successful implementations of ESC and P\&O to various processes, see, e.g., \cite{Yilmaz2023,Marko2023}, but the widespread application of these solutions in industry is still pending. An underlying implementation limitation is that ESC and P\&O involve continuous perturbation of the control parameters, even after the optimum has been found. This leads to perturbation around the optimum, which reduces performance, and increases the wear of machinery, resource use and operating cost unnecessarily.

Approaches that address unnecessary perturbations in model-free optimization often involve reducing the size of the perturbations based on some estimate of the proximity of the optimum. For example, in \cite{Buyukdegirmenci2010}, the step size of P\&O is adapted based on changes in the measured function value, while \cite{Atta2016} scales the perturbation magnitude based on the gradient estimate in ESC. Similarly, \cite{Moura2013,Bafandeh2017} use a Lyapunov-based switching scheme to estimate how close the system is to the optimum and reduce the perturbation in ESC accordingly. A disadvantage of these approaches is that if the perturbation amplitude is reduced to zero, a change in the location of the optimum that does not change the function value at the previous optimum is not detected. Another way to reduce unnecessary perturbations, which is commonly used in P\&O for maximum power point tracking, is to use other parameters, such as the incremental conductance, based on the underlying physical model \cite{Hussein1995,Femia2007,Zhang2009a}. However, this approach is limited to this specific application, and it requires system knowledge that is not always available. In \cite{Weekers2025a}, some optimization steps are made based on kernel-based cost function approximations, but this method does not consider time-varying cost functions.

Although there exist many different approaches to reduce unnecessary perturbations in model-free optimization, a method that effectively tracks time-varying optima without any a-priori model knowledge remains underdeveloped. In this paper, we present \textit{a computationally efficient data-based optimization approach, for which convergence to and tracking of a time-varying optimum can be proven}. In this so-called uncertainty-based perturb and observe (uP\&O), cost function value estimates and their uncertainties are modeled based on prior measurements, and the input (to be optimized) at each time step is selected based on these estimates. By reusing older measurements while increasing their uncertainties, it is possible to track time-varying optima without needing to perturb constantly. These methodological features are key in industrial applications. The main contribution consists of the following aspects:
\begin{itemize}
	\item A new uncertainty-based perturb and observe approach is presented.
	\item A formal analysis shows that the developed uP\&O method converges to and tracks the time-varying optimum under mild conditions.
	\item Simulation results demonstrate that uP\&O uses fewer perturbations to accurately track a time-varying optima compared to standard P\&O. The method is shown to perform similarly to three alternative methods for which convergence is difficult to prove.
\end{itemize}

A preliminary version of the uP\&O approach was presented in \cite{Aarnoudse2025a}, in which the input for each time step follows from optimizing future inputs over a fixed horizon. Additional contributions of the current paper w.r.t. \cite{Aarnoudse2025a} are increased flexibility in modeling the uncertainty, a new method to select the inputs for each time step that is more computationally efficient, and a formal convergence analysis.

This paper is organized as follows. In Section \ref{sec:problem}, the problem is introduced. In Section \ref{sec:approach}, a method is introduced to model time-varying function values and their uncertainties based on measurements, and in Section \ref{sec:inputs}, a method to select suitable inputs that balance exploration (i.e., testing if the optimum is changing) and exploitation (i.e., staying at the current estimate of the optimum) is presented. Convergence conditions for uP\&O are developed in Section \ref{sec:convergence}. In Section \ref{sec:sims}, the method is validated in a simulation-based case study, and finally, conclusions are given in Section \ref{sec:conclusions}.

\textit{Notation:} The sets of real and natural numbers are denoted by $\mathbb{R} $ and $\mathbb{N}$, respectively. Subscripts generally denote discrete time instants (e.g., $u_k$ denotes the value of $u$ at discrete time instant $k$), while superscripts in parentheses denote a member of a discrete set of values (e.g., $u^{(i)}$ denotes the $i$-th element of a set $\{u^{(1)},u^{(2)},...\}$). The ceiling function is denoted by $\lceil x \rceil$, which rounds a number $x \in \mathbb{R}$ to the closest integer larger than $x$.

\section{Problem formulation} \label{sec:problem}

Consider the problem of finding and tracking the maximum (or minimum) of an unknown, single-input single-output (SISO), time-varying function $f_k(u) \: \in \mathbb{R}$, using measurements $y_k \in \mathbb{R}$ of its function value at time index $k\in \mathbb{N}$ for an input $u \in \mathcal{U} \subset \mathbb{R}$. The equidistantly spaced input set is given by $\mathcal{U} = \{u\in \mathbb{R}: (\exists i \in \mathbb{Z})[u=i\Delta_u]\}$ for some positive constant $\Delta_u \in \mathbb{R}_{>0}$. For ease of notation, define $u^{(i)} = i \Delta_u$ for all $i \in \mathbb{Z}$.  Only one measurement of the function can be made at each time instance, and the relation between a function measurement at an input $u_k \in \mathcal{U}$ and the corresponding function value at that input is given by
	\begin{align} \label{eq:yk}
		y_k = f_k (u_k) + \rho \varepsilon_k,
	\end{align}
where $\varepsilon_k \in \mathbb{R}$ is an i.i.d., white noise variable, and $\rho \in \mathbb{R}_{>0}$ is a positive scaling constant. The term $\rho \varepsilon_k$ in \eqref{eq:yk} captures the error between the function value and the corresponding measurement. For later notational convenience, we denote the index of the input by $\iota_k \in \mathbb{Z}$, i.e., $u_k = \iota_k \Delta_u$. Noting that $u^{(i)} = i \Delta_u$ for all $i \in \mathbb{Z}$, we can also write $u_k = u^{(\iota_k)}$. We adopt the following assumption regarding the maximum of the function $f_k(u)$.
\begin{ass} \label{assum:maximum}
	At each time $k \in \mathbb{N}$, there exists a unique maximizer $u_k^* = \argmax_{u\in \mathcal{U}} f_k(u)$.
\end{ass}
Similar to $u_k$, we denote the index of $u_k^*$ by $\iota_k^* \in \mathbb{Z}$ and may write $u_k^* = u^{(\iota_k^*)} = \iota_k^* \Delta_u$.

A common approach to find and track the maximizer $u_k^*$ of an unknown, static, time-varying function is perturb and observe. The main idea is to perturb the input and determine the direction of improvement from the corresponding measured output. The direction of improvement $g_k \in\{-1,1\} $ is determined from two subsequent measurements as
	\begin{align} g_{k+1}=
		\begin{cases}
			g_k, & \text{if }  y_k\geq y_{k-1}, \\
			-g_k & \text{if } y_k < y_{k-1}.
		\end{cases} \label{eq:gk}
	\end{align}
This formula is used to update the input according to
	\begin{align} \label{eq:uk}
		u_{k+1} = u_k + g_{k+1} \Delta_u.
	\end{align}
Thus, the direction in which the input is updated remains constant if the function value increases, and if the function value reduces, the direction is reversed. Note that throughout this paper, maximization is considered, but all algorithms can easily be adapted for minimization by selecting the next input as the one expected to decrease (instead of increase) the function value.
	
P\&O is capable of finding and tracking maxima effectively, but it has two main disadvantages. First, it is sensitive to noise on the measurements $y_k$, which can cause the algorithm to move in the wrong direction. Second, after the optimum is found, the algorithm continues perturbing the input. Although perturbations are needed to keep track of the changing maximum of a time-varying cost function, these continuing, unconditional perturbations are often excessive and may reduce the overall performance, lead to increased wear and tear of the physical system, or result in unnecessary resource usage. This paper aims to develop a computationally efficient perturbation-based optimization method that meets the following requirements:
\begin{enumerate}[R1]
	\item Find and track a time-varying maximum,
	\item in the presence of noise,
	\item while limiting the number of perturbations,
	\item in a computationally efficient manner,
	\item with convergence guarantees.
\end{enumerate}

The method consists of three main parts: in Section \ref{sec:approach}, a model is introduced that is used to obtain function value estimates based on prior measurements while taking into account both noise (R2) and time variation of the cost function. In Section \ref{sec:inputs}, this model is used to determine the input for the next time step (R1 and R3) in a computationally efficient manner (R4). In Section \ref{sec:convergence}, it is shown under what conditions the presented method is guaranteed to converge to a bounded region around the optimum (R5).

\section{Uncertainty-based perturb and observe} \label{sec:approach}

In this section, a new uncertainty-based perturb-and-observe (uP\&O) algorithm is introduced that meets requirements R1-R5. First, a model for the time-varying cost function $f_k(u)$ is introduced. Second, the best function value estimates that follow from solving this model are computed. In Section \ref{sec:inputs}, this model is used to select an input for the next time step that balances exploration and exploitation, in the sense that it considers the expected value and uncertainties of the estimates.

\subsection{Uncertainty-based model} \label{subsec:model}

Consider the function $f_k(u)$ with noisy measurements according to \eqref{eq:yk}.  Since the function changes over time, older measurements give potentially outdated and uncertain information about the current function value at the same input. The relation between a current function value $f_k(u^{(i)})$, at the current time $k$, and older measurements $y_j$, taken at past time $j$ and associated to the same input, is modeled using an uncertainty that increases over time:
\begin{align} \label{eq:measurment_model}
	f_k(u^{(i)}) = y_j - \frac{\rho}{\sqrt{\omega_{j,k}}} \varepsilon_j, \quad \mbox{if $u_j = u^{(i)}$},
\end{align}
for all $k,j \in \mathbb{N}$ such that $k \geq j$. To derive the method, it is now assumed that $\varepsilon_j$ is a standard, Gaussian variable, and $\rho\in \mathbb{R}_{>0}$ is the same positive scaling constant as in \eqref{eq:yk}. The time-dependent weighting $\omega_{j,k} \in (0,1]$ is used to artificially inflate the noise of older measurements, so that, in effect,  older measurements contain less information about the current function value than newer measurements. This is done to take into account any time-variations of the function without explicitly modeling those. The weights should be close to 1 for small values of $k-j$ (recent measurements), and reduce to 0 for increasing values of $k-j$ (older measurements). To that end, we take 
\begin{equation} \label{eq:omega}
	\omega_{j,k} = \sum_{q = 0}^M \zeta_{j,k,q}
\end{equation}
with
\begin{equation} \label{eq:zeta}
	\zeta_{j,k,q} = \frac{1}{q!} \left(\ln\left(\frac{1}{\lambda}\right)(k - j)\right)^q \lambda^{k-j}
\end{equation}
for some nonnegative integer $M \in \mathbb{N}$ and some positive constant $\lambda \in (0,1)$. The forgetting factor $\lambda$ determines the rate at which the weight $\omega_{j,k}$ converges from one to zero for constant $k$ and decreasing $j$. Choosing $\lambda$ closer to zero leads to a faster increase in uncertainty, while $\lambda$ closer to one causes the uncertainty to increase more slowly. The term $M$ determines how long the weight remains close to one. In particular, for $M=0$ the uncertainty increases with a factor $\frac{1}{\lambda}$ for each time step, while for larger values of $M$ the increase in uncertainty varies over time, as illustrated in Fig. \ref{fig:weights}. In general, faster increases in uncertainty lead to more frequent perturbations in the input selection algorithm developed in Section \ref{sec:inputs}. 

\begin{figure}[t]
	\centering
	\setlength\figureheight{.4\figurewidth}
	\includegraphics{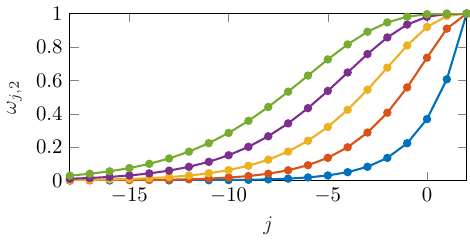}
	\caption{Example of the weights $\omega_{j,k}$ as a function of $j$ for $\lambda = e^{-0.5}$, $k=2$ and $M=0$ (\protect\blueline), $M=1$ (\protect\redline), $M=2$ (\protect\yelline), $M=3$ (\protect\purline) and $M=4$ (\protect\greenline). \label{fig:weights}}
\end{figure}

\subsection{Optimal estimates using past data: model solution}
To predict the function value $f_{k+1}^{(i)}$ for input value $u^{(i)}$ at the next time step $k+1$, consider the model:
\begin{align} \label{eq:model_prediction}
	\mathcal{\hat{M}}_{k+1}^{(i)} = \left\{ \tilde{f}_{k+1}^{(i)} = y_j - \frac{\rho}{\sqrt{\omega_{j,k+1}}} \varepsilon_j, \, \forall j \in \mathcal{J}_{k}(u^{(i)}) \right\},
\end{align}
which captures all available information about the function value at $u^{(i)}$, and predicts the future function value $f_{k+1}\left(u^{(i)}\right)$. The index set $\mathcal{J}_k^{(i)}$ contains all time indices up to time $k$ for which $u_j = u^{(i)}$, and is defined as
\begin{equation} \label{eq:J}
	\mathcal{J}_k^{(i)} = \left\{j \in \{0,1,\dots,k\} : \, u_j = u^{(i)} \right\},
\end{equation}
for all $k\in \mathbb{N}$ and $i \in \mathbb{Z}$. The best estimate of the function value $f_{k+1}^{(i)}$ given the equations in the model $\hat{\mathcal{M}}_{k+1}^{(i)}$ in \eqref{eq:model_prediction} follows from solving the model for the variable $\tilde{f}_{k+1}^{(i)}$. Assuming that the set $\mathcal{J}^{(i)}_{k}$ is nonempty, the solution of the model $\hat{\mathcal{M}}_{k+1}^{(i)}$ in \eqref{eq:model_prediction} for $\tilde{f}^{(i)}_{k+1}$ is given by
\begin{align} \label{eq:solution_prediction}
	\tilde{f}^{(i)}_{k+1|\hat{\mathcal{M}}_{k+1}^{(i)}} = \frac{\sum_{j \in \mathcal{J}_k^{(i)}} \omega_{j,k+1} y_j}{\sum_{j \in \mathcal{J}_k^{(i)}} \omega_{j,k+1} } - \rho \frac{\sum_{j \in \mathcal{J}_k^{(i)}} \sqrt{\omega_{j,k+1}} \varepsilon_j}{\sum_{j \in \mathcal{J}_k^{(i)}} \omega_{j,k+1} },
\end{align}
for all $i \in \{1,2,\dots,N\}$. The notation $\tilde{f}^{(i)}_{k+1|\hat{\mathcal{M}}_{k+1}^{(i)}}$ indicates that this is a solution based on the model $\hat{\mathcal{M}}_{k+1}^{(i)}$ for the next time step $k+1$. 

\begin{rem}
	If the set $\mathcal{J}^{(i)}_{k}$ is empty, then the solution for $\tilde{f}^{(i)}_{k+1}$ is unknown. 	
\end{rem} 

Let the mean and variance of $\tilde{f}^{(i)}_{k+1|\hat{\mathcal{M}}_{k+1}^{(i)}}$ in \eqref{eq:solution_prediction} be denoted by 
\begin{align}
	\hat{\mu}_{k+1}^{(i)} &= \mathbb{E}\left[\tilde{f}^{(i)}_{k+1|\hat{\mathcal{M}}_{k+1}^{(i)}}\right], \\
	\hat{\Sigma}_{k+1}^{(i)} &= \mathbb{E}\left[\left(\tilde{f}^{(i)}_{k+1|\hat{\mathcal{M}}_{k+1}^{(i)}} -\mathbb{E}\left[\tilde{f}^{(i)}_{k+1|\hat{\mathcal{M}}_{k+1}^{(i)}}\right] \right)^2\right].
\end{align}
It follows that
\begin{align} \label{eq:mean}
	\hat{\mu}_{k+1}^{(i)} &= \frac{\sum_{j \in \mathcal{J}_k^{(i)}} \omega_{j,k+1} y_j}{\sum_{j \in \mathcal{J}_{k}^{(i)}} \omega_{j,k+1} } = \frac{\sum_{q = 0}^M\sum_{j \in \mathcal{J}_{k}^{(i)}} \zeta_{j,k+1,q} y_j}{\sum_{q = 0}^M\sum_{j \in \mathcal{J}_{k}^{(i)}} \zeta_{j,k+1,q} } \\ \nonumber &= \frac{\sum_{q = 0}^M \xi_{k+1,q}^{(i)}}{\sum_{q = 0}^M \phi_{k+1,q}^{(i)} }
\end{align}
and
\begin{align} \label{eq:variance}
	\hat{\Sigma}_{k+1}^{(i)} &=  \frac{\rho^2}{\sum_{j \in \mathcal{J}_{k}^{(i)}} \omega_{j,k+1} } =  \frac{\rho^2}{\sum_{q = 0}^M\sum_{j \in \mathcal{J}_{k}^{(i)}} \zeta_{j,k+1,q} } \\ \nonumber &=  \frac{\rho^2}{\sum_{q = 0}^M \phi_{k+1,q}^{(i)}}
\end{align}
with
\begin{align} \label{eq:xi}
	\xi_{k+1,q}^{(i)} &= \sum_{j \in \mathcal{J}_{k}^{(i)}} \zeta_{j,k+1,q} y_j, \\ 
	\phi_{k+1,q}^{(i)} &= \sum_{j \in \mathcal{J}_{k}^{(i)}} \zeta_{j,k+1,q}. \label{eq:phi}
\end{align}
Note that $\hat{\mu}_{k+1}^{(i)}$ is the maximum likelihood estimate of $f_{k+1}(u^{(i)})$ given the model and the measurement data from time zero to time $k$, and $\hat{\Sigma}_{k+1}^{(i)}$ is the corresponding variance of the estimate. Both the function value estimate and its variance are used in Section \ref{sec:inputs} to select the next input value. Although the expressions in \eqref{eq:mean} and \eqref{eq:variance} are easy to compute for low values of $k$, computations are troublesome for high values of $k$ due to the potentially large number of terms in the sums in \eqref{eq:xi} and \eqref{eq:phi}. Therefore, the aim is to construct recursive expressions for $\xi_{k+1,q}^{(i)}$ in \eqref{eq:xi} and $\phi_{k+1,q}^{(i)}$ in \eqref{eq:phi} that are computationally efficient to update, to enable fast online reconstruction of $\hat{\mu}_{k+1}^{(i)}$ and $\hat{\Sigma}_{k+1}^{(i)}$.

Consider the vector notation
\begin{align}
	\boldsymbol\xi_{k}^{(i)} = \begin{bmatrix}
		\xi_{k,0}^{(i)} \\ \xi_{k,1}^{(i)} \\ \vdots \\ \xi_{k,M}^{(i)}
	\end{bmatrix}, \quad \boldsymbol\phi_k^{(i)} = \begin{bmatrix}
		\phi_{k,0}^{(i)} \\ \phi_{k,1}^{(i)} \\ \vdots \\ \phi_{k,M}^{(i)}
	\end{bmatrix}.
\end{align} 

Then, $\hat{\mu}_{k+1}^{(i)}$ and $\hat{\Sigma}_{k+1}^{(i)}$ can be rewritten as shown in the following lemma.
\begin{lem} \label{lem:mean_variance}
	The expressions for $\hat{\mu}_{k+1}^{(i)}$ and $\hat{\Sigma}_{k+1}^{(i)}$ in respectively \eqref{eq:mean} and \eqref{eq:variance} are equivalent to
	\begin{align} \label{eq:mu_update_2}
		\hat{\mu}_{k+1}^{(i)}&= \frac{\mathbf{c}\tran \mathbf{A} \boldsymbol\xi_{k}^{(i)}}{\mathbf{c}\tran \mathbf{A} \boldsymbol\phi_{k}^{(i)}}, \\ 
		\hat{\Sigma}_{k+1}^{(i)} &= \frac{\rho^2}{\mathbf{c}\tran \mathbf{A} \boldsymbol\phi_{k+1}^{(i)}}, \label{eq:sig_update_2}
	\end{align}
	and $\boldsymbol\xi_{k}^{(i)}$ and $\boldsymbol\phi_{k}^{(i)}$ follow from solving
	\begin{align} \label{eq:xi_update}
		\boldsymbol\xi_{k+1}^{(i)} &= \mathbf{A} \boldsymbol\xi_k^{(i)} + \mathbf{b} v^{\xi,(i)}_{k+1} \\ \label{eq:phi_update}
		\boldsymbol\phi_{k+1}^{(i)} &= \mathbf{A} \boldsymbol\phi_k^{(i)} + \mathbf{b} v^{\phi,(i)}_{k+1},
	\end{align}
	for the signals
	\begin{align}
		v^{\xi,(i)}_{k+1} &= \begin{cases}
			y_{k+1}, & \mbox{if $u_{k+1} = u^{(i)}$}\\
			0, & \mbox{otherwise}\\
		\end{cases}  \\
		v^{\phi,(i)}_{k+1} &= \begin{cases}
			1,  &\mbox{if $u_{k+1} = u^{(i)}$}\\
			0,  &\mbox{otherwise.}\\
		\end{cases} 
	\end{align}
	Both $\boldsymbol\xi$ and $\boldsymbol\phi$ are initialized at $\mathbf{0}$ at $k=-1$. The Toeplitz matrix $\mathbf{A}$ in \eqref{eq:mu_update_2}-\eqref{eq:phi_update} is given by
	\begin{align}
		&\mathbf{A} = \\ \nonumber &\begin{bmatrix}
			\lambda & 0 & 0 & \cdots & 0\\
			\lambda \ln\left( \frac{1}{\lambda}\right) & \lambda & 0 & \cdots & 0\\
			\frac{\lambda}{2} \left(\ln\left( \frac{1}{\lambda}\right)\right)^2 & \lambda \ln\left( \frac{1}{\lambda}\right) & \lambda & \ddots & \vdots\\
			\vdots & \ddots & \ddots & \ddots & 0 \\
			\frac{\lambda}{M!} \left(\ln\left( \frac{1}{\lambda}\right)\right)^M & \cdots & \frac{\lambda}{2} \left(\ln\left( \frac{1}{\lambda}\right)\right)^2 & \lambda \ln\left( \frac{1}{\lambda}\right) & \lambda
		\end{bmatrix},
	\end{align}
	and $\mathbf{b}$ and $\mathbf{c}$ are vectors given by
	\begin{align}
		\mathbf{b} &= \begin{bmatrix}
			1 & 0 & 0 & \dots & 0
		\end{bmatrix}\tran, \\
		\mathbf{c} &= \begin{bmatrix}
			1 & 1 & \dots & 1
		\end{bmatrix}. \tran
	\end{align}
\end{lem}

Proof: The proof is given in Appendix \ref{app:proof_mean_variance}. $\hfill\blacksquare$ 

Lemma \ref{lem:mean_variance} enables fast updating of the function value estimates and uncertainties at every time step, to comply with requirement R4. In addition, the amount of information that needs to be stored is small. Note that the values of $\boldsymbol\xi_{k}^{(i)}$ and $\boldsymbol\phi_{k}^{(i)}$ can be computed for any $k \in \mathbb{N}$ by setting $\boldsymbol\xi_{-1}^{(i)} = \boldsymbol\phi_{-1}^{(i)} = \mathbf{0}$ for all $i \in \mathbb{Z}$ and using the update equations in \eqref{eq:xi_update} and \eqref{eq:phi_update} to compute the values for $k$. Moreover, by doing so, the values for $\boldsymbol\xi_{k}^{(i)}$ and $\boldsymbol\phi_{k}^{(i)}$ are defined for all $k \in \mathbb{N}$ and all $i \in \mathbb{Z}$, even if the function value for a given input $u^{(i)}$ has not been measured (i.e., $\mathcal{J}_k^{(i)}$ is an empty set). In the next section, the model is used to determine a suitable input value for the next time step to support optimization and choose whether to perturb or not.

\section{Input selection} \label{sec:inputs}

In this section, an input for the next time step is selected in a way that balances exploration and exploitation to meet requirements R1 and R3, i.e., such that a potentially time-varying optimum is tracked without continuously perturbing the system unnecessarily. To this end, the relation between the function $f$ at the current input and at the two neighboring input values is modeled based on the uncertain function value estimates. In addition, the complete method is summarized, and alternative methods for input selection are discussed.

\subsection{Modeling the relation between three function values} \label{subsec:3pointmodel}

Next, the aim is to determine the input values $u_k$ for all $k \in \mathbb{N}$ such that maximization is achieved. The algorithm requires an initialization procedure where the first input $u_0$ and second input $u_1$ are chosen. Without any information about the location of the maximizer, we can choose any $u_0 \in \mathcal{U}$. Subsequently, the second input $u_1$ is chosen as a neighboring input value, i.e., $u_1 \in \{u_0+\Delta_u, u_0 - \Delta_u\} \subset \mathcal{U}$. After $u_0$ and $u_1$ have been chosen, any subsequent input value $u_{k+1}$ is selected using the uncertainty-based perturb and observe method for all $k \geq 1$. This input value is either the same as the previous input value or a neighboring point, i.e., $u_{k+1} \in \{u_k - \Delta_u, u_k, u_k + \Delta_u\}$ for all $k \geq 1$. 

To determine the next input value, the model of Section \ref{sec:approach} that predicts the function values at the next time step is used to construct a second, \textit{local}, model for the function. This local model describes the relation between the current input and its two neighbors. The model is piece-wise linear, and the individual pieces are based on the estimated mean values and uncertainties described in Section \ref{sec:approach}. The minimizer of this local model around the current input is used as the next input.

Consider the model \eqref{eq:model_prediction} with solution $\tilde{f}^{(i)}_{k+1|\hat{\mathcal{M}}_{k+1}^{(i)}}$ according to \eqref{eq:solution_prediction}. This solution can be written in the following form, in the sense that the probability distribution on both sides of the equation is the same:
\begin{align} \label{eq:solution_prediction_rewritten}
	\tilde{f}^{(i)}_{k+1|\hat{\mathcal{M}}_{k+1}^{(i)}}  = \hat{\mu}_{k+1}^{(i)} + \sqrt{\hat{\Sigma}^{(i)}_{k+1}} \varepsilon_{k+1}^{(i)},
\end{align}
where $\varepsilon_{k+1}^{(i)} \in \mathbb{R}$ is a standard, Gaussian variable.

A local model is built around the current input $u_k = u^{(\iota_k)}$ and its two neighboring points $u^{(\iota_k-1)}$ and $u^{(\iota_k+1)}$ for time step $k+1$, using the estimates of the function values $f_{k+1}(u)$ that follow from the model solution \eqref{eq:solution_prediction_rewritten}. The local model is piecewise-linear, and is based on the following a priori model:
\begin{align} \label{eq:pw_linear_model}
	\tilde{f}^{(\iota_k-1)}_{k+1} - 2\tilde{f}^{(\iota_k)}_{k+1} + \tilde{f}^{(\iota_k+1)}_{k+1} = \delta e^{(\iota_k)}_{k+1},
\end{align}
where $\delta \in \mathbb{R}_{>0}$ is a positive constant and $e^{(\iota_k)}_k$ is a standard, Gaussian variable. Note that the left-hand side of this model is zero if and only if the estimated function values $\tilde{f}^{(\iota_k-1)}_{k+1}$, $\tilde{f}^{(\iota_k)}_{k+1}$, and $\tilde{f}^{(\iota_k+1)}_{k+1}$ are on a straight line, since the distances between subsequent points are the same, i.e., $u^{(\iota_k+1)} - u^{(\iota_k)} = u^{(\iota_k)} - u^{(\iota_k-1)} = \Delta_u$. In addition, $\tilde{f}^{(\iota_k-1)}_{k+1}$, $\tilde{f}^{(\iota_k)}_{k+1}$, and $\tilde{f}^{(\iota_k+1)}_{k+1}$ in \eqref{eq:pw_linear_model} should correspond to the solution of the model $\mathcal{\hat{M}}_{k+1}^{(i)}$. This leads to the following local model for the function $f_{k+1}(u)$ around the point $u^{(\iota_k)}$, based on \eqref{eq:solution_prediction_rewritten} and \eqref{eq:pw_linear_model}:

\begin{align} \label{eq:model_S}
	\mathcal{S}_{k+1}^{(\iota_k)} = \left\{ \begin{aligned}
		& \tilde{f}^{(\iota_k-1)}_{k+1} = \hat{\mu}_{k+1}^{(\iota_k-1)} + \sqrt{\hat{\Sigma}^{(\iota_k-1)}_{k+1}} \varepsilon_{k+1}^{(\iota_k-1)}\\
		& \tilde{f}^{(\iota_k)}_{k+1} = \hat{\mu}_{k+1}^{(\iota_k)} + \sqrt{\hat{\Sigma}^{(\iota_k)}_{k+1}} \varepsilon_{k+1}^{(\iota_k)}\\
		& \tilde{f}^{(\iota_k+1)}_{k+1} = \hat{\mu}_{k+1}^{(\iota_k+1)} + \sqrt{\hat{\Sigma}^{(\iota_k+1)}_{k+1}} \varepsilon_{k+1}^{(\iota_k+1)}\\
		& \tilde{f}^{(\iota_k-1)}_{k+1} - 2\tilde{f}^{(\iota_k)}_{k+1} + \tilde{f}^{(\iota_k+1)}_{k+1} = \delta e^{(\iota_k)}_{k+1}
	\end{aligned} \right\}.
\end{align}
The intuitive interpretation of design variable $\delta$ is that it determines how much the piecewise linear model deviates from a linear model. For low values of $\delta$, the local model is forced to be closer to linear than for higher values of $\delta$, which cause the function to bend to stay closer to the estimates $\hat{\mu}_{k+1}^{(i-1)}$, $\hat{\mu}_{k+1}^{(i)}$ and $\hat{\mu}_{k+1}^{(i+1)}$. For the input selection in Section \ref{subsec:input}, lower values of $\delta$ imply more frequent perturbations. 

The mean of the solution of this model is used to determine the next input, and is given in the following lemma. 

\begin{lem} \label{lem:model_solution}
	The mean of the solution of the model $\mathcal{S}_{k+1}^{(\iota_k)}$ in \eqref{eq:model_S} for the variables $\tilde{f}^{(\iota_k-1)}_{k+1}$, $\tilde{f}^{(\iota_k)}_{k+1}$, and $\tilde{f}^{(\iota_k+1)}_{k+1}$ is given by
	\begin{align} \label{eq:h_solution}
		&\begin{bmatrix}
			h^{(\iota_k-1)}_{k+1} \\ h^{(\iota_k)}_{k+1} \\ h^{(\iota_k+1)}_{k+1} 
		\end{bmatrix} := \mathbb{E} \left(
		\begin{bmatrix}
			\tilde{f}^{(\iota_k-1)}_{k+1|\mathcal{S}_{k+1}^{(\iota_k)}} \\ \tilde{f}^{(\iota_k)}_{k+1|\mathcal{S}_{k+1}^{(\iota_k)}} \\ \tilde{f}^{(\iota_k+1)}_{k+1|\mathcal{S}_{k+1}^{(\iota_k)}}
		\end{bmatrix} \right) =	\begin{bmatrix}
		\frac{\theta_{k}^{(\iota_k-1)}}{\zeta_k} \\ \frac{\theta_{k}^{(\iota_k)}}{\zeta_k} \\ \frac{\theta_{k}^{(\iota_k+1)}}{\zeta_k}
		\end{bmatrix} 
	\end{align}
	with 
	\begin{align}
			&\theta_{k}^{(\iota_k-1)} = \mathbf{c}\tran \mathbf{A} \boldsymbol\xi_k^{(\iota_k-1)} \left(\nu^2 \mathbf{c}\tran \mathbf{A} \boldsymbol\phi_k^{(\iota_k)} \mathbf{c}\tran \mathbf{A} \boldsymbol\phi_k^{(\iota_k+1)} + \right. \\ \nonumber & \quad \left.  \mathbf{c}\tran \mathbf{A} \boldsymbol\phi_k^{(\iota_k)} + 4 \mathbf{c}\tran \mathbf{A} \boldsymbol\phi_k^{(\iota_k+1)}  \right) +   2 \mathbf{c}\tran \mathbf{A} \boldsymbol\xi_k^{(\iota_k)} \mathbf{c}\tran \mathbf{A} \boldsymbol\phi_k^{(\iota_k+1)}  \\ \nonumber
			& \quad  - \mathbf{c}\tran \mathbf{A} \boldsymbol\xi_k^{(\iota_k+1)} \mathbf{c}\tran \mathbf{A} \boldsymbol\phi_k^{(\iota_k)} ,\\ 
			&\theta_{k}^{(\iota_k)} = \mathbf{c}\tran \mathbf{A} \boldsymbol\xi_k^{(\iota_k)} \left( \nu^2 \mathbf{c}\tran \mathbf{A} \boldsymbol\phi_k^{(\iota_k-1)} \mathbf{c}\tran \mathbf{A} \boldsymbol\phi_k^{(\iota_k+1)} + \right. \\ \nonumber & \quad \left.    \mathbf{c}\tran \mathbf{A} \boldsymbol\phi_k^{(\iota_k-1)} + \mathbf{c}\tran \mathbf{A} \boldsymbol\phi_k^{(\iota_k+1)}  \right) + 2   \mathbf{c}\tran \mathbf{A} \boldsymbol\xi_k^{(\iota_k-1)} \mathbf{c}\tran \mathbf{A} \boldsymbol\phi_k^{(\iota_k+1)} \\ \nonumber
			&  \quad + \mathbf{c}\tran \mathbf{A} \boldsymbol\xi_k^{(\iota_k+1)} \mathbf{c}\tran \mathbf{A} \boldsymbol\phi_k^{(\iota_k-1)} ,\\
			&\theta_{k}^{(\iota_k+1)} = \mathbf{c}\tran \mathbf{A} \boldsymbol\xi_k^{(\iota_k+1)} \left(\nu^2 \mathbf{c}\tran \mathbf{A} \boldsymbol\phi_k^{(\iota_k-1)} \mathbf{c}\tran \mathbf{A} \boldsymbol\phi_k^{(\iota_k)} + \right. \\ \nonumber & \quad \left.  
			 4 \mathbf{c}\tran \mathbf{A} \boldsymbol\phi_k^{(\iota_k-1)} + \mathbf{c}\tran \mathbf{A} \boldsymbol\phi_k^{(\iota_k)} \right) -\mathbf{c}\tran \mathbf{A} \boldsymbol\xi_k^{(\iota_k-1)} \mathbf{c}\tran \mathbf{A} \boldsymbol\phi_k^{(\iota_k)} \\ \nonumber
			&  \quad  + 2 \mathbf{c}\tran \mathbf{A} \boldsymbol\xi_k^{(\iota_k)} \mathbf{c}\tran \mathbf{A} \boldsymbol\phi_k^{(\iota_k-1)} , \\
			&\zeta_k = \nu^2 \mathbf{c}^T \mathbf{A} \boldsymbol\phi_k^{(\iota_k-1)} \mathbf{c}^T \mathbf{A} \boldsymbol\phi_k^{(\iota_k)} \mathbf{c}^T \mathbf{A} \boldsymbol\phi_k^{(\iota_k+1)} \\ \nonumber & \quad +   \mathbf{c}^T \mathbf{A} \boldsymbol\phi_k^{(\iota_k)} \mathbf{c}^T \mathbf{A} \boldsymbol\phi_k^{(\iota_k+1)} + 4 \mathbf{c}^T \mathbf{A} \boldsymbol\phi_k^{(\iota_k-1)} \mathbf{c}^T \mathbf{A} \boldsymbol\phi_k^{(\iota_k+1)} \\ \nonumber & \quad + \mathbf{c}^T \mathbf{A} \boldsymbol\phi_k^{(\iota_k-1)} \mathbf{c}^T \mathbf{A} \boldsymbol\phi_k^{(\iota_k)},
	\end{align}
	and 
	\begin{align}
		\nu = \frac{\delta}{\rho}.
	\end{align}
\end{lem}
Proof: The proof is given in Appendix \ref{app:proof_model_solution}. $\hfill\blacksquare$ 

The three (mean) points $h^{(\iota_k-1)}_{k+1}$, $h^{(\iota_k)}_{k+1}$, and $h^{(\iota_k+1)}_{k+1}$ form an estimate of the function $f_{k+1}(u)$ at the points $u^{(\iota_k-1)}$, $u^{(\iota_k)}$, and $u^{(\iota_k+1)}$, respectively. Note that for any given $k \geq 1$, at least two of the three possible input values $u_k - \Delta_u$, $u_k$, and $u_k + \Delta_u$ for the next input $u_{k+1}$ have been measured previously, which implies that at least two of the three values $\mathbf{c}^T \mathbf{A} \boldsymbol\phi_k^{(\iota_k-1)}$, $\mathbf{c}^T \mathbf{A} \boldsymbol\phi_k^{(\iota_k)}$, and $\mathbf{c}^T \mathbf{A} \boldsymbol\phi_k^{(\iota_k+1)}$ are positive (in fact, $\mathbf{c}^T \mathbf{A} \boldsymbol\phi_k^{(\iota_k)}$ is always positive because the function value at $u_k = u^{(\iota_k)}$ has been measured at time $k$). In turn, it follows that $\zeta_k$ has a positive value for all $k \geq 1$ and that the values of $h^{(\iota_k-1)}_{k+1}$, $h^{(\iota_k)}_{k+1}$ and $h^{(\iota_k+1)}_{k+1}$ in \eqref{eq:h_solution} are well defined.

Using the equations in \eqref{eq:mean} and \eqref{eq:variance}, we can write \eqref{eq:h_solution} as
{\begin{align} \label{eq:h_three_points}
	&\begin{bmatrix}
		h^{(\iota_k-1)}_{k+1} \\ h^{(\iota_k)}_{k+1} \\ h^{(\iota_k+1)}_{k+1}
	\end{bmatrix} = \\ \nonumber
	& \begin{cases} \! \!
		\begin{bmatrix}
			\hat{\mu}_{k+1}^{(\iota_k-1)} \\ \hat{\mu}_{k+1}^{(\iota_k)} \\ 2 \hat{\mu}_{k+1}^{(\iota_k)} - \hat{\mu}_{k+1}^{(\iota_k-1)}
		\end{bmatrix}, \quad \mbox{if $\mathcal{J}_k^{(\iota_k+1)} = \emptyset$,}\\
	\!\!	\begin{bmatrix}
			2 \hat{\mu}_{k+1}^{(\iota_k)} - \hat{\mu}_{k+1}^{(\iota_k+1)} \\ \hat{\mu}_{k+1}^{(\iota_k)} \\ \hat{\mu}_{k+1}^{(\iota_k+1)} 
		\end{bmatrix},  \quad \mbox{if $\mathcal{J}_k^{(\iota_k-1)} = \emptyset$,}\\
	\!\!	 \begin{bmatrix}
			\hat{\mu}_{k+1}^{(\iota_k-1)} \\ \hat{\mu}_{k+1}^{(\iota_k)} \\ \hat{\mu}_{k+1}^{(\iota_k+1)}
		\end{bmatrix} \!-\! \frac{\hat{\mu}_{k+1}^{(\iota_k-1)} - 2 \hat{\mu}_{k+1}^{(\iota_k)}  + \hat{\mu}_{k+1}^{(\iota_k+1)} }{1 + \frac{\hat{\Sigma}_{k+1}^{(\iota_k-1)}}{(\nu\rho)^2}  + 4 \frac{\hat{\Sigma}_{k+1}^{(\iota_k)}}{(\nu\rho)^2} + \frac{\hat{\Sigma}_{k+1}^{(\iota_k+1)}}{(\nu\rho)^2}} \! \begin{bmatrix}
			\frac{\hat{\Sigma}_{k+1}^{(\iota_k-1)}}{(\nu\rho)^2} \\  \frac{-2 \hat{\Sigma}_{k+1}^{(\iota_k)}}{(\nu\rho)^2} \\ \frac{\hat{\Sigma}_{k+1}^{(\iota_k+1)}}{(\nu\rho)^2}
		\end{bmatrix}\!\!, {\text{otherwise}}.
	\end{cases}	
\end{align}}
The first condition, $\mathcal{J}_k^{(\iota_k+1)} = \emptyset$, holds if the point $u^{(\iota_k+1)}$ has not been measured yet, the second condition, $\mathcal{J}_k^{(\iota_k-1)} = \emptyset$,  holds if the point $u^{(\iota_k-1)}$ has not been measured yet, and the third condition holds if all three points have been measured before. It is evident from \eqref{eq:h_three_points} that $h^{(\iota_k-1)}_{k+1} \approx \hat{\mu}_{k+1}^{(\iota_k-1)}$, $h^{(\iota_k)}_{k+1} \approx \hat{\mu}_{k+1}^{(\iota_k)}$, and $h^{(\iota_k+1)}_{k+1} \approx \hat{\mu}_{k+1}^{(\iota_k+1)}$ if all three points have been measured before and if the value of the tuning parameter $\nu$ is sufficiently large. Similarly, if all three points have been measured before and if the value $\nu$ is sufficiently small (i.e., it is close to zero), then the values of $h^{(\iota_k-1)}_{k+1}$, $h^{(\iota_k)}_{k+1}$, and $h^{(\iota_k+1)}_{k+1}$ are almost on a straight line because $h^{(\iota_k-1)}_{k+1} -2 h^{(\iota_k)}_{k+1} + h^{(\iota_k+1)}_{k+1} \approx 0$ in that case. If only two points have been measured (i.e., the first two conditions in \eqref{eq:h_three_points}), we always have that $h^{(\iota_k-1)}_{k+1}$, $h^{(\iota_k)}_{k+1}$, and $h^{(\iota_k+1)}_{k+1}$ are on a straight line since $h^{(\iota_k-1)}_{k+1} -2 h^{(\iota_k)}_{k+1} + h^{(\iota_k+1)}_{k+1} = 0$.

\subsection{Optimal input selection and implementation} \label{subsec:input}

Having defined $h^{(\iota_k-1)}_{k+1}$, $h^{(\iota_k)}_{k+1}$, and $h^{(\iota_k+1)}_{k+1}$, we propose the following uncertainty-based perturb and observe algorithm for the selection of the next input:
\begin{equation} \label{eq:optimization1}
	u_{k+1} = u^{(\iota_{k+1})},
\end{equation}
with
\begin{align} \label{eq:optimization2}
	&\iota_{k+1} \in \\ \nonumber & \begin{cases}
		 \iota_k-1, \: \mbox{if $p_k^{(\iota_k-1)} < p_k^{(\iota_k+1)}$ and $0 \leq h_{k+1}^{(\iota_k)} - h_{k+1}^{(\iota_k+1)} \leq \tau$,}\\
		 \iota_k+1,  \: \mbox{if $p_k^{(\iota_k-1)} > p_k^{(\iota_k+1)}$ and $0 \leq h_{k+1}^{(\iota_k)} - h_{k+1}^{(\iota_k-1)} \leq \tau$,}\\
		 \argmax_{i\in\{\iota_k-1,\iota_k,\iota_k+1\}} h_{k+1}^{(i)},  \quad \mbox{otherwise,}
	\end{cases} 
\end{align}
for some small constant $\tau \in \mathbb{R}_{>0}$, where $p_k^{(i)}$ is defined as
\begin{equation} \label{eq:p_k}
	p_k^{(i)} = \max\left\{ \mathcal{J}_k^{(i)}, -1  \right\}
\end{equation}
for all $i\in\{\iota_k-1,\iota_k,\iota_k+1\}$. Hence, $p_k^{(i)}$ is the last time that input value $u^{(i)}$ was measured, or minus one if that point has not been measured. Note that $p_k^{(\iota_k)} = k$, and that $ -1\leq p_k^{(\iota_k-1)} < p_k^{(\iota_k)}$ and $ -1\leq p_k^{(\iota_k+1)} < p_k^{(\iota_k)}$ for all $k \in \mathbb{N}$.

The first two cases in \eqref{eq:optimization2} correspond to the situation where the point $h_{k+1}^{(\iota_k)}$ of the current input is only slightly higher than that of the other input that was most recently measured, i.e., $h_{k+1}^{(\iota_k-1)}$ or $h_{k+1}^{(\iota_k+1)}$, depending on whether $p_k^{(\iota_k-1)}$ or $p_k^{(\iota_k+1)}$ is larger. In this situation, the algorithm is forced to perturb to the least recently measured point. In all other cases, the next input corresponds to the highest of the three points $h_{k+1}^{(\iota_k-1)}$, $h_{k+1}^{(\iota_k)}$ and $h_{k+1}^{(\iota_k+1)}$ in the local model. The resulting uncertainty-based perturb and observe approach is summarized in Algorithm \ref{alg:upo} for $n_{\text{iteration}}$ iterations, and illustrated in Fig. \ref{fig:ex_dec}.

\begin{figure}
	\centering
	\setlength\figurewidth{.9\defcolwidth}
	\setlength\figureheight{.7\figurewidth}
	\includegraphics{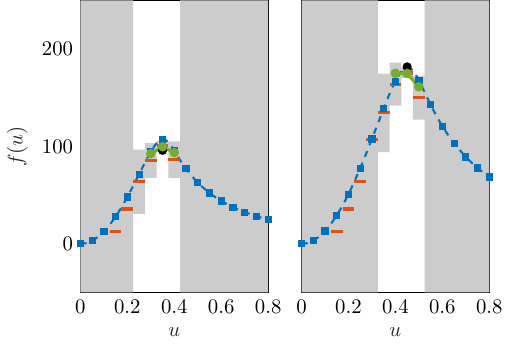} 
	\caption{Illustration of the decision process of uncertainty-based perturb and observe at $k=50$ (left) and $k=95$ (right) for a time-varying function (\protect\bluedash) (see Section \ref{sec:sims} for simulation details, cf. the uP\&O result in Fig. \ref{fig:input}). The next input is chosen based on the highest point in the local model (\protect\grefulldot\protect\greenline) around the current input (\protect\blackfulldot), which is based on the function value estimates (\protect\redline) and their uncertainties (indicated in gray). At $k=50$ (left), the current input is the highest point, and the input remains constant. At $k=95$ (right), the point on the left is the highest in the local model, causing a perturbation.  \label{fig:ex_dec}}
\end{figure}

\begin{algorithm}[H]
	\caption{Uncertainty-based perturb and observe} 	\label{alg:upo}
	\begin{algorithmic}[1]
		\State{Initialize $u_0$ and $u_1 \in \{u_0 -\Delta_u,u_0 + \Delta_u\}$, and set $n_{\text{iteration}}$, $\rho$, $\nu$, $\tau$, $\lambda$ and $M$}
		\State{Measure $y_k = f_k(u_k)+ \rho \varepsilon_k$ for $k=0$}
		\State{\textbf{for} $k=1:n_{\text{iteration}}$}
		\State{\quad Measure $y_k = f_k(u_k)+ \rho \varepsilon_k$ and compute $	\boldsymbol\xi_{k+1}^{(i)}$ and}
		\Statex{\quad $\boldsymbol\phi_{k+1}^{(i)}$, using respectively \eqref{eq:xi_update} and \eqref{eq:phi_update} in Lemma \ref{lem:mean_variance}.}
		\State{\quad Compute $h^{(\iota_k-1)}_{k+1}$, $h^{(\iota_k)}_{k+1}$ and $h^{(\iota_k+1)}_{k+1}$ according to}
		\Statex{\quad Lemma \ref{lem:model_solution}.}
		\State{\quad Compute $p^{(\iota_k-1)}_k$, $p^{(\iota_k)}_k$ and $p^{(\iota_k+1)}_k$ using \eqref{eq:p_k}.}
		\State{\quad Determine which of the three cases in \eqref{eq:optimization2} applies}
		\Statex{\quad and take $u_{k+1} = u^{(\iota_{k+1})}$ accordingly, see \eqref{eq:optimization1}.}	
		\State{\textbf{end} }	
	\end{algorithmic}
\end{algorithm}

\subsection{Comparison to other input selection methods} \label{subsec:other_inputs}
\label{subsec:sampling} 

There exist other methods to select the next input based on the model developed in Section \ref{sec:approach}. The preliminary results in \cite{Aarnoudse2025a} consider a special case of the weighting used in the current paper with $M = 0$. Then, the optimal input over a finite horizon of for example five time steps is computed. While this approach performs well in simulations, only convergence conditions for standard P\&O are provided, and this method does not enable tuning for convergence without resorting to standard P\&O.

Other input selection methods that could be considered include Highest Expected Improvement \cite{Jones1998} and Thompson sampling \cite{Thompson1933,Chapelle2011}, both of which are often used in the context of Gaussian processes and bandit problems. When using Highest Expected Improvement for optimization, the index of the next input is chosen as 
\begin{align}
	&\iota_{k+1} \in \arg \max_{j \in \{\iota_{k}-1,\iota_{k},\iota_{k}+1\}} \left( \hat{\mu}^{(j)}_{k+1} - \hat{\mu}^{(\iota_{k})}_{k+1} - \alpha \right) \times \\ \nonumber & \Phi \left( \frac{\hat{\mu}^{(j)}_{k+1} - \hat{\mu}^{(\iota_{k})}_{k+1} - \alpha}{\sqrt{\hat{\Sigma}^{(j)}_{k+1}}} \right) + \sqrt{\hat{\Sigma}^{(j)}_{k+1}} \varphi \left( \frac{\hat{\mu}^{(j)}_{k+1} - \hat{\mu}^{(\iota_{k})}_{k+1} - \alpha}{\sqrt{\hat{\Sigma}^{(j)}_{k+1}}} \right),
\end{align} 
where the hyperparameter $\alpha$ determines the trade-off between exploration and exploitation and $\Phi$ and $\varphi$ denote respectively the cumulative density function and the probability density function of a standard Gaussian variable. In Thompson sampling, one sample is taken of each of the three Gaussian distributions with mean $\hat{\mu}^{(j)}_{k+1}$ and variance $\hat{\Sigma}^{(j)}_{k+1}$ for $j \in \{\iota_{k}-1,\iota_{k},\iota_{k}+1\}$, and the index that gives the best sample is used as the next input. Thus, the sampling is random with probabilities proportional to the expected values and uncertainties of the considered inputs. Highest Expected Improvement and Thompson sampling may perform well in simulations, but due to the stochastic components it is difficult to prove that the input converges to the optimal input and tracks it over time in a deterministic sense. Therefore, these methods are not preferred for the perturb-and-observe setting considered here. Still, in Section \ref{sec:sims}, these approaches are compared in simulation to the approach presented in this paper.

\section{Convergence of uncertainty-based P\&O} \label{sec:convergence}

In this section, convergence conditions for uncertainty-based perturb and observe are presented to demonstrate that the algorithm meets requirement R5. To build towards this result, first some preliminary assumptions and lemmas are introduced. Second, the convergence of uncertainty-based perturb and observe for time-varying functions is analyzed. 

\subsection{Assumptions and lemmas to prove convergence}	

First, a different assumption on the measurement uncertainty is required.

\begin{ass} \label{assum:noise}
	The noise variable $\varepsilon_k$ in \eqref{eq:yk} is bounded and satisfies $|\varepsilon_k|\leq 1 \: \forall k \in \mathbb{N}$. 
\end{ass}

Assumption \ref{assum:noise} is more restrictive than the description of $\varepsilon_k$ as an i.i.d., white, standard Gaussian variable used to derive the approach in Section \ref{sec:approach}. The assumption is necessary to accommodate a deterministic convergence analysis, and is not limiting in practice, since measurement disturbances are often bounded with, e.g., truncated Gaussian distributions. In addition, the convergence analysis considers the deterministic outcome, i.e., a choice of input, of the model in Section \ref{sec:approach}, regardless of the underlying statistical assumptions or possible inaccuracies in the model. Finally, the effective magnitude of the noise is determined by $\rho$, even when $\varepsilon_k$ is bounded.

\noindent In addition, the following assumptions are adopted.

\begin{ass} \label{assum:curvature}
	The increment of $f_k(u)$ between two neighboring inputs at the same time instance is bounded as follows:
	\begin{equation}
		L_b\left|u^{(i+\frac{1}{2})} - u_k^*\right|^2  \leq \left(f_k(u^{(i)}) - f_k(u^{(i+1)}) \right) \left(u^{(i+\frac{1}{2})} - u_k^*\right),
	\end{equation}
	for all $k \in \mathbb{N}$, all $i \in \mathbb{Z}$, and some positive constant $L_b \in \mathbb{R}_{>0}$, where $u^{(i+\frac{1}{2})}$ is defined as
	\begin{equation}
		u^{(i+\frac{1}{2})} = \frac{u^{(i)} + u^{(i+1)}}{2} = \left(i + \frac{1}{2}\right) \Delta_u.
	\end{equation}
\end{ass}
Note that Assumption \ref{assum:curvature} implies that the function becomes increasingly steeper further away from the maximizer.
\begin{ass} \label{assum:time_change}
	The rate of change of $f_k(u)$ between two subsequent time steps, for constant $u$, is bounded, i.e., there exists a positive constant $L_k \in \mathbb{R}_{>0}$ such that
	\begin{equation}
		|f_{k+1}(u) - f_k(u)| \leq L_k
	\end{equation}
	for all $k \in \mathbb{N}$ and all $u \in \mathcal{U}$.
\end{ass}

Assumption \ref{assum:curvature} is a common assumption in the analysis of extremum-seeking control algorithms, and enables the analysis of the convergence of the algorithm to a bounded region around the optimum. This is relevant when the cost landscape around the optimum is almost flat, which in combination with noisy measurements can make guaranteed convergence to the exact optimum impossible. Assumption \ref{assum:time_change} ensures that the rate of change of the time-varying function is not so fast that the algorithm cannot track it. 

The convergence proof of the presented uncertainty-based perturb and observe algorithm is presented in three parts. First, we show in Lemma~\ref{lem:maximizer} that the change of the optimum $u_k^*$ of the function $f_k$ in a finite amount of time is bounded under the given assumptions. Second, we show in Lemma~\ref{lem:convergence_bounds} that the input $u_k$ will converge to a region of the optimum if the used algorithm satisfies certain conditions. Third, we show in Theorem \ref{thm:convergence} that the presented algorithms satisfy these conditions if tuned correctly, which implies that the convergence bounds of Lemma~\ref{lem:convergence_bounds} hold.

\begin{lem} \label{lem:maximizer}
	Under Assumptions~\ref{assum:maximum}-\ref{assum:time_change}, there exists a positive constant $L^* \in \mathbb{R}_{>0}$, such that
	\begin{equation} \label{eq:lem_maximizer}
		|u_{k+N}^* - u_k^*| \leq  L^* \sqrt{N} 
	\end{equation}
	for all $N \in \mathbb{N}$.
\end{lem}
\begin{proof}
	See Appendix~\ref{app:proof_lem_maximizer}.
\end{proof}

\begin{lem} \label{lem:convergence_bounds}
	Under Assumptions~\ref{assum:maximum}-\ref{assum:time_change}, if $|u_{k+1} - u_k| \leq \Delta_u$ for all $k \in \mathbb{N}$, and if there exists a positive constant $b \in \mathbb{R}_{>0}$, such that, whenever $|u_k - u_{k+1}^*| \geq b$, we have that $|u_{k+1} - u_{k+1}^*| = |u_k - u_{k+1}^*| - \Delta_u$ (i.e., $u_{k+1}$ is one step closer to $u_{k+1}^*$ than $u_k$) for any $k \geq k_0$, where $k_0 \in \mathbb{N}$ is a nonnegative integer, then there exist positive constants $c_1, c_2 \in \mathbb{R}_{>0}$, such that
	\begin{equation} \label{eq:lem_sup1}
		\sup_{k \geq 0} |u_k - u_k^*| \leq |u_0-u_0^*| + c_1
	\end{equation}
	and
	\begin{equation} \label{eq:lem_lim_sup1}
		\limsup_{k\rightarrow\infty} |u_k - u_k^*| \leq c_2
	\end{equation}
	for all $u_0 \in \mathcal{U}$.
\end{lem}
\begin{proof}
	See Appendix~\ref{app:proof_lem_convergence_bounds}.
\end{proof}

\subsection{Convergence of uP\&O}

For the convergence of the presented uncertainty-based perturb and observe algorithm, we obtain the following result.

\begin{thm} \label{thm:convergence}
	Let $\nu^* \in \mathbb{R}_{>0}$ be any positive constant. Under Assumptions~\ref{assum:noise}-\ref{assum:time_change}, using the uncertainty-based perturb and observe optimizer in \eqref{eq:optimization1}-\eqref{eq:optimization2} with $M = 0$ or $M = 1$, there exist constants $c_1, c_2 \in \mathbb{R}_{>0}$, and a sufficiently small constant $\lambda^*\in (0,1)$, such that
	\begin{equation} \label{eq:lem_sup}
		\sup_{k \geq 0} |u_k - u_k^*| \leq |u_0-u_0^*| + c_1
	\end{equation}
	and
	\begin{equation} \label{eq:lem_lim_sup}
		\limsup_{k\rightarrow\infty} |u_k - u_k^*| \leq c_2
	\end{equation}
	for all $u_0 \in \mathcal{U}$, all $u_1 \in \{u_0 - \Delta_u,u_0 + \Delta_u \}$,  all $\nu \in (0,\nu^*]$, and all $\lambda \in (0,\lambda^*]$.
\end{thm}
\begin{proof}
The proof is given in Appendix \ref{app:proof_thm}	
\end{proof} 

We care to note that in simulation, the algorithm converges also for larger values of $M$. However, since the equations for the bounds become cumbersome for $M>1$, no formal proof is given here. Theorem \ref{thm:convergence} shows that the uP\&O algorithm can be designed such that convergence to a bounded region around the optimal input is guaranteed, even if the function that is maximized is time-varying. A similar analysis holds when the algorithm is used for minimization. Note that $c_1$ and $c_2$ in respectively \eqref{eq:lem_sup} and \eqref{eq:lem_lim_sup} are related to the function properties in Assumption \ref{assum:curvature} and \ref{assum:time_change}. They decrease for increasing $L_u$ (steeper function away from the maximizer) and decreasing $L_k$ (smaller variation over time). 

\begin{rem}
	The proof of Theorem \ref{thm:convergence} is based on the limit case where $\lambda^*$ is chosen so that the algorithm perturbs at every time step. The actual number of perturbations needed to track the maximum depends on how fast the function varies over time, and is related to $L_k$ in Assumption \ref{assum:time_change}. This potential reduction in the number of perturbations is the main advantage of uP\&O compared to P\&O, and can lead to significant performance improvements, e.g., reduced cost as less time is spend away from the maximum (i.e., away from optimal performance), as is illustrated in the next section.
\end{rem}

\section{Illustrative example} \label{sec:sims}

In this section, the presented method is applied to a case study of the optimization of a photo-voltaic solar array. First, the problem setup is introduced, and second, simulation results are provided for the presented uP\&O method, standard P\&O, and the alternative input selection methods from Section \ref{subsec:other_inputs}.

\subsection{Problem setup}

\begin{figure}[t]
	\centering
	\setlength\figureheight{.4\figurewidth}
	\includegraphics{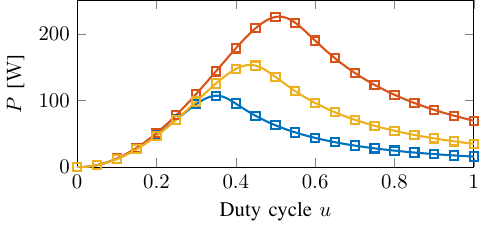}
	\caption{In the simulation example, function $f_k(u)$ is the steady-state produced power $P$ of a photovoltaic array as a function of the duty cycle $u$, shown at time indices $k=50$ (\protect\blueline), $k=150$ (\protect\redline) and $k=250$ (\protect\yelline). The squares indicate the possible inputs. \label{fig:graph}}
\end{figure}

\begin{table}[t]
	\centering
	\caption{Parameters for the photovoltaic model. \label{tab:pars}}
	\begin{tabular}{|l|l|} \hline
		Parameter & Description \\ \hline
		$T_r =$ \SI{298.15}{\kelvin} & Reference temperature \\	
		$I_s =$ \SI{5.61}{\ampere} & Reference short-circuit current at $T_r$ \\
		$I_0 =$ \SI{1.13e-6}{\ampere} & Nominal reverse saturation current at $T_r$ \\	
		$k_i =$ \SI{1.96e-3}{\ampere/ \kelvin} & Short-circuit current temperature coefficient \\
		$N =$ 1.81 & Ideality factor \\
		$E_g =$ \SI{1.16}{\electronvolt} & Band gap energy for silicon \\
		$k =$ \SI{1.38e-23}{\joule / \kelvin} & Boltzmann constant \\
		$q =$ \SI{1.60e-19}{\coulomb} & Charge of an electron \\
		$n_s =$ 72 & Number of photo-voltaic cells in series \\
		$R_s =$ \SI{2.83e-3}{\ohm} & Series resistance \\
		$R_p =$ \SI{8.7}{\ohm} & Parallel resistance \\ 
		$C_c =$ \SI{1}{\milli \farad} & Converter's capacitance \\
		$L_c =$ \SI{5}{\milli \henry} & Converter's inductance \\
		$R_c =$ \SI{2}{\ohm} & Converter's resistance \\		
		\hline
	\end{tabular}
\end{table}

Consider a photo-voltaic array with $n_s$ cells in series, which is affected by varying sunlight and temperature conditions throughout a day. Based on the model from \cite{Vachtsevanos1987}, the light-generated current $i_s$ and the reverse saturation current $i_0$ are given by
\begin{align} \label{eq:is}
	i_s &= \left(I_s + k_i(T-T_r)\right) \frac{S}{1000} \\
	i_0 &= I_0 \left(\frac{T}{T_r}\right)^3 e^{\frac{E_g}{NV_t} \left(\frac{T}{T_r} -1 \right)}, \quad V_t = \frac{kT}{q}, \label{eq:io}
\end{align}
for a photovoltaic cell with temperature $T$ and solar irradiance $S$. All parameters are given in Table \ref{tab:pars}. The output current $i$ of an array with $n_s$ cells in series is given by
\begin{align}
	i = i_s - i_0\left( e^{\frac{v+iR_s n_s}{N V_t n_s}}-1\right) - \frac{v+i R_s n_s}{R_p n_s}. \label{eq:i}
\end{align}
Here, $v$ denotes the output voltage. A dc-dc buck converter is used to connect each array to a dc load. The converter dynamics are represented by the model in \cite{Li2016} as 
\begin{align}
	C_c \dot{v} &= i(v;T,S)-i_L u \\
	L_c \dot{i}_L &= -i_L R_c + vu, \label{eq:doti}
\end{align}
with $i(v;T,S)$ the nonlinear mapping \eqref{eq:i} from the output voltage $v$ to the output current. Using \eqref{eq:is}-\eqref{eq:doti}, the nonlinear steady-state mapping between the duty cycle $u$ and the produced power $P$ can be computed for any combination of temperature and solar irradiance.

\begin{figure}[t]
	\centering
	\setlength\figurewidth{.8\defcolwidth}
	\setlength\figureheight{.36\defcolwidth}
	\includegraphics{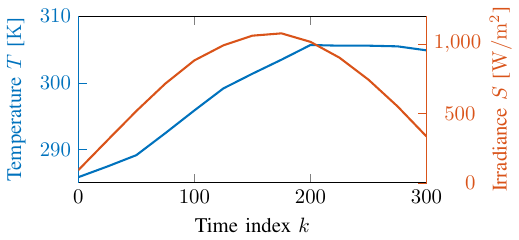}
	\caption{Temperature (\protect\blueline) and solar irradiance (\protect\redline) between 6AM ($k=0$) and 6PM ($k=300$) on a clear and sunny day. \label{fig:temp}}
\end{figure}

Fig. \ref{fig:temp} shows possible values for the temperature and irradiance between 6AM (corresponding to $k=0$) and 6PM (corresponding to $k=300$) on a clear day. The varying temperature and irradiance values lead to variations in the mapping from the duty cycle $u$ to the produced power, as illustrated for $k=50$, $k=150$ and $k=250$ in Fig. \ref{fig:graph}. The aim is to find and track the optimal duty cycle that maximizes the produced power throughout the day. It is assumed that noisy measurements of the produced power are available, where the noise is zero-mean Gaussian white noise with a standard deviation of $5$. All simulations use the same noise realization.

\subsection{Comparison between standard and uncertainty-based P\&O}

To compare the performance of the presented uncertainty-based perturb and observe to standard perturb and observe, both methods are applied to the problem of finding and tracking the duty cycle that maximizes the energy output for the solar array. For uP\&O, the parameters are chosen as $\lambda = e^{-0.5}$, $\nu = 3$, $M=1$ and $\rho = 5$. In Fig. \ref{fig:input}, it is shown that uP\&O requires far fewer perturbations to track the time-varying optimal input compared to P\&O, spending only 92 instead of 172 of the 300 time steps away from the optimal input. Fig. \ref{fig:cost} shows how this reduction in perturbations affects the total power generated by the solar array. The presented uP\&O outperforms P\&O by 2.5\%, and generates 7.8\% more power compared to using the best constant duty cycle. If the optimal duty cycle depicted in Fig. \ref{fig:input} is used, the performance is only 1.8\% better than that of uP\&O.

\begin{figure}[t]
	\centering
	\setlength\figureheight{.4\defcolwidth}
	\includegraphics{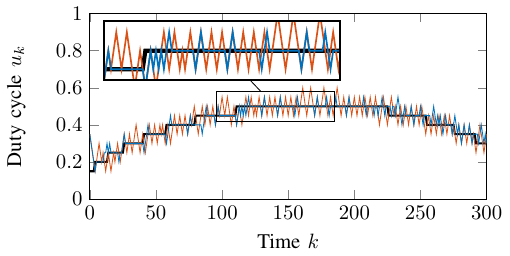}
	\caption{Uncertainty-based P\&O (\protect\blueline) uses far fewer perturbations than standard P\&O (\protect\redline) to track the time-varying optimal input (\protect\blackline). \label{fig:input}}
\end{figure}

\begin{figure}[t]
	\centering
	\setlength\figureheight{.4\defcolwidth}
	\includegraphics{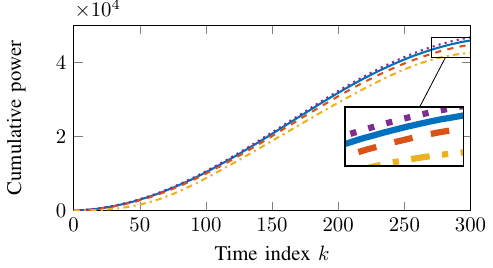}
	\caption{Using uncertainty-based P\&O (\protect\blueline) leads to a 2.5\% increase in total power output over a day compared to standard P\&O (\protect\reddash). The uP\&O output is 7.8\% higher than the maximum power output using a constant duty cycle (\protect\yeldashdot), and only 1.8\% lower than the maximum output if the optimal duty cycle is always known exactly (\protect\purdots).  \label{fig:cost}}
\end{figure}

\subsection{Comparison to other sampling methods}

The main advantage of the uncertainty-based perturb-and-observe method presented in this paper is its ability to reduce the number of perturbations compared to P\&O, while maintaining fast tracking and convergence properties. Other sampling methods, such as the previous uP\&O method from \cite{Aarnoudse2025a} that is based on computing the optimal input over a fixed horizon, highest expected improvement and Thompson sampling, see Section \ref{subsec:sampling} for more details, are capable of fast tracking with reduced perturbations, but similar convergence properties cannot be derived easily. Still, it is useful to compare these methods in simulation.

\begin{figure}[t]
	\centering
	\setlength\figureheight{.4\defcolwidth}
	\includegraphics{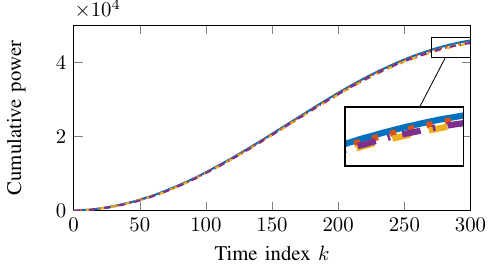}
	\caption{The presented uncertainty-based P\&O method (\protect\blueline) outperforms the uP\&O method from \cite{Aarnoudse2025a} (\protect\reddots), as well as highest expected improvement sampling (\protect\yeldash) and Thompson sampling (\protect\purdashdot), but the differences in performance are small.  \label{fig:cost2}}
\end{figure}

All methods use the model from \cite{Aarnoudse2025a}, which is equivalent to the model $\mathcal{M}$ in Section \ref{sec:approach} for $M=0$. In the uP\&O method from \cite{Aarnoudse2025a}, $\lambda = 0.88$ is used, while the highest expected improvement and Thompson sampling methods use $\lambda = 0.95$. In addition, highest expected improvement sampling is implemented using $\alpha = 0.0001$. The results in Fig. \ref{fig:cost2} show that the differences in performance for these methods are small. In terms of the number of perturbations, the uP\&O method of \cite{Aarnoudse2025a}, highest expected improvement and Thompson sampling spend, respectively, 122, 144 and 149 time steps away from the optimal input, compared to 92 for the uP\&O method presented here. While these numbers depend strongly on tuning, the results demonstrate that all tested sampling methods lead to improvements compared to standard P\&O. The main advantage of the presented uP\&O method is that convergence can be guaranteed, as shown in Section \ref{sec:convergence}.

\section{Conclusions} \label{sec:conclusions}

A new uncertainty-based perturb-and-observe method is presented that significantly reduces the number of perturbations needed to track time-varying optima compared to existing P\&O methods, while also exhibiting guaranteed convergence properties. This improves the industrial applicability of model-free optimization methods that can improve the performance of highly uncertain, time-varying systems, for which model-based optimization solutions are often too conservative. The method employs a model of the time-varying function based on measurements, the uncertainty of which increases over time. Based on this model, a second, local model of the performance cost at the current input and its neighboring inputs is constructed, which balances between exploration and exploitation while selecting the input for the next time step. The model can be updated in a computationally efficient manner through recursive expressions, even when a large number of previous measurements is used. A convergence analysis shows that, under mild conditions, uP\&O converges to a bounded region around the time-varying optimum and remains there. In addition, simulation results demonstrate that uP\&O can outperform standard P\&O and reduces the number of perturbations significantly. Future work includes experimental implementation, inclusion of feedforward based on measurements or limited model knowledge, and the case of multiple inputs.

\appendices

\section{Proof of Lemma \ref{lem:mean_variance}} \label{app:proof_mean_variance}

\begin{proof}[of Lemma \ref{lem:mean_variance}]
First, it follows from \eqref{eq:zeta} that
\begin{align}
	&\zeta_{j,k+1,q} = \frac{1}{q!} \left(\ln\left(\frac{1}{\lambda}\right)(k + 1 - j)\right)^q \lambda^{k + 1-j}\\ \nonumber
	&= \frac{\lambda}{q!} \left(\ln\left(\frac{1}{\lambda}\right)\right)^q \sum_{r = 0}^q \binom{q}{r} (k - j)^r \lambda^{k-j}\\ \nonumber
	&=\sum_{r = 0}^q \frac{\lambda}{(q-r)!} \left(\ln\left(\frac{1}{\lambda}\right)\right)^{q-r}   \frac{1}{r!} \left(\ln\left(\frac{1}{\lambda}\right)(k - j)\right)^r \lambda^{k-j}\\ \nonumber
	&= \sum_{r = 0}^q \frac{\lambda}{(q-r)!} \left(\ln\left(\frac{1}{\lambda}\right)\right)^{q-r}   \zeta_{j,k,r}.
\end{align}
Now, using the definitions of $\xi_{k,q}^{(i)}$ in \eqref{eq:xi} and $\phi_{k,q}^{(i)}$ in \eqref{eq:phi}, and the definitions of the set $\mathcal{J}_k^{(i)}$ in \eqref{eq:J} and $\zeta_{j,k,q}$ in \eqref{eq:zeta},  we obtain that
\begin{align} \label{eq:proof_lem_mv1}
		&\xi_{k+1,q}^{(i)} = \sum_{j \in \mathcal{J}_{k+1}^{(i)}} \zeta_{j,k+1,q} y_j \\ \nonumber
		&= \sum_{j \in \mathcal{J}_{k}^{(i)}} \zeta_{j,k+1,q} y_j + 
		\begin{cases}
			y_{k+1}, & \mbox{if $u_{k+1} = u^{(i)}$ and $q = 0$}\\ 
			0, & \mbox{otherwise.}
		\end{cases}
\end{align} 
Note that the sets $\mathcal{J}_{k+1}^{(i)}$ and $\mathcal{J}_{k}^{(i)}$ are equal if $u_{k+1} \neq u^{(i)}$, since there is no new measurement at $u^{(i)}$ in that case. In addition, if $q\neq 0$, $\zeta_{j,k,q} = 0$ for $j=k$, see \eqref{eq:zeta}. The first part in the right-hand side of \eqref{eq:proof_lem_mv1} can be rewritten to
\begin{align} \nonumber
			\sum_{j \in \mathcal{J}_{k}^{(i)}} \zeta_{j,k+1,q} y_j &= \sum_{r = 0}^q \frac{\lambda}{(q-r)!} \left(\ln\left(\frac{1}{\lambda}\right)\right)^{q-r} \sum_{j \in \mathcal{J}_{k}^{(i)}}  \zeta_{j,k,r} y_j \\ 
			&= \sum_{r = 0}^q \frac{\lambda}{(q-r)!} \left(\ln\left(\frac{1}{\lambda}\right)\right)^{q-r} \xi_{k,r}^{(i)} 
		\end{align}
and, similarly,
\begin{align} \nonumber
		\phi_{k+1,q}^{(i)} =  & \sum_{j \in \mathcal{J}_{k+1}^{(i)}} \zeta_{j,k+1,q} 
		=  \sum_{r = 0}^q \frac{\lambda}{(q-r)!} \left(\ln\left(\frac{1}{\lambda}\right)\right)^{q-r} \phi_{k,r}^{(i)}  \\  
		&+ \begin{cases}
			1, & \mbox{if $u_{k+1} = u^{(i)}$ and $q = 0$}\\
			0, & \mbox{otherwise}
		\end{cases}.
\end{align}
Hence, both $\xi_{k+1,q}^{(i)}$ and $\phi_{k+1,q}^{(i)}$ can be expressed as linear functions of $\xi_{k,r}^{(i)}$ (and $y_{k+1}$) and $\phi_{k,r}^{(i)}$, respectively, for $r \leq q$, and these functions can be rewritten to \eqref{eq:xi_update} and \eqref{eq:phi_update}.
\end{proof}

\section{Proof of Lemma \ref{lem:model_solution}} \label{app:proof_model_solution}

\begin{proof}[of Lemma \ref{lem:model_solution}]
	The minimum-variance solution of the model $\mathcal{S}_{k+1}^{(\iota_k)}$ in \eqref{eq:model_S} for the variables $\tilde{f}^{(\iota_k-1)}_{k+1}$, $\tilde{f}^{(\iota_k)}_{k+1}$, and $\tilde{f}^{(\iota_k+1)}_{k+1}$ is given by
	\begin{align} \nonumber
		&\begin{bmatrix}
			\tilde{f}^{(\iota_k-1)}_{k+1|\mathcal{S}_{k+1}^{(\iota_k)}} \\ \tilde{f}^{(\iota_k)}_{k+1|\mathcal{S}_{k+1}^{(\iota_k)}} \\ \tilde{f}^{(\iota_k+1)}_{k+1|\mathcal{S}_{k+1}^{(\iota_k)}}
		\end{bmatrix} = \begin{bmatrix}
			\hat{\mu}_{k+1}^{(\iota_k-1)} \\ \hat{\mu}_{k+1}^{(\iota_k)} \\ \hat{\mu}_{k+1}^{(\iota_k+1)}
		\end{bmatrix} - 
		\frac{\hat{\mu}_{k+1}^{(\iota_k-1)} - 2\hat{\mu}_{k+1}^{(\iota_k)} + \hat{\mu}_{k+1}^{(\iota_k+1)}}{1 + \frac{\hat{\Sigma}_{k+1}^{(\iota_k-1)}}{\delta^2} + 4\frac{\hat{\Sigma}_{k+1}^{(\iota_k)}}{\delta^2} + \frac{\hat{\Sigma}_{k+1}^{(\iota_k+1)}}{\delta^2}} \times \\ \nonumber 
		& \quad \begin{bmatrix}
			\frac{\hat{\Sigma}_{k+1}^{(\iota_k-1)}}{\delta^2} \\ -2\frac{\hat{\Sigma}_{k+1}^{(\iota_k)}}{\delta^2} \\ \frac{\hat{\Sigma}_{k+1}^{(\iota_k+1)}}{\delta^2}
		\end{bmatrix} 
		- \delta  \begin{bmatrix}
			\frac{\sqrt{\hat{\Sigma}_{k+1}^{(\iota_k-1)}}}{\delta} \varepsilon_{k+1}^{(\iota_k-1)} + \frac{\hat{\Sigma}_{k+1}^{(\iota_k-1)}}{\delta^2} e^{(\iota_k)}_{k+1} \\ \frac{\sqrt{\hat{\Sigma}_{k+1}^{(\iota_k)}}}{\delta} \varepsilon_{k+1}^{(\iota_k)} - 2\frac{\hat{\Sigma}_{k+1}^{(\iota_k)}}{\delta^2} e^{(\iota_k)}_{k+1} \\ \frac{\sqrt{\hat{\Sigma}_{k+1}^{(\iota_k+1)}}}{\delta} \varepsilon_{k+1}^{(\iota_k+1)} + \frac{\hat{\Sigma}_{k+1}^{(\iota_k+1)}}{\delta^2} e^{(\iota_k)}_{k+1} \end{bmatrix}\\ \nonumber
		& \quad + \delta \left( \frac{ \frac{\sqrt{\hat{\Sigma}_{k+1}^{(\iota_k)}}}{\delta} \varepsilon_k^{(\iota_k)} - 2  \frac{\sqrt{\hat{\Sigma}_{k+1}^{(\iota_k)}}}{\delta} \varepsilon_{k+1}^{(\iota_k)} +  \frac{\sqrt{\hat{\Sigma}_{k+1}^{(\iota_k+1)}}}{\delta} \varepsilon_k^{(\iota_k+1)}}{1 + \frac{\hat{\Sigma}_{k+1}^{(\iota_k-1)}}{\delta^2} + 4\frac{\hat{\Sigma}_{k+1}^{(\iota_k)}}{\delta^2} + \frac{\hat{\Sigma}_{k+1}^{(\iota_k+1)}}{\delta^2}} \right. \\ 
		& \quad	+\left. \frac{ \left( \frac{\hat{\Sigma}_{k+1}^{(\iota_k-1)}}{\delta^2} + 4\frac{\hat{\Sigma}_{k+1}^{(\iota_k)}}{\delta^2} + \frac{\hat{\Sigma}_{k+1}^{(\iota_k+1)}}{\delta^2}\right) e^{(\iota_k)}_k}{1 + \frac{\hat{\Sigma}_{k+1}^{(\iota_k-1)}}{\delta^2} + 4\frac{\hat{\Sigma}_{k+1}^{(\iota_k)}}{\delta^2} + \frac{\hat{\Sigma}_{k+1}^{(\iota_k+1)}}{\delta^2}} \right)
		\begin{bmatrix}
			\frac{\hat{\Sigma}_{k+1}^{(\iota_k-1)}}{\delta^2} \\ -2\frac{\hat{\Sigma}_{k+1}^{(\iota_k)}}{\delta^2} \\ \frac{\hat{\Sigma}_{k+1}^{(\iota_k+1)}}{\delta^2}
		\end{bmatrix}
		.
	\end{align}
	The mean of this solution is derived as follows: 
	{\small \begin{align}
			& \mathbb{E}  \begin{bmatrix}
				\tilde{f}^{(\iota_k-1)}_{k+1|\mathcal{S}_{k+1}^{(\iota_k)}} \\ \tilde{f}^{(\iota_k)}_{k+1|\mathcal{S}_{k+1}^{(\iota_k)}} \\ \tilde{f}^{(\iota_k+1)}_{k+1|\mathcal{S}_{k+1}^{(\iota_k)}}
			\end{bmatrix} = \\ \nonumber &=  \begin{bmatrix}
				\hat{\mu}_{k+1}^{(\iota_k-1)} \\ \hat{\mu}_{k+1}^{(\iota_k)} \\ \hat{\mu}_{k+1}^{(\iota_k+1)}
			\end{bmatrix}   - \frac{\hat{\mu}_{k+1}^{(\iota_k-1)} - 2\hat{\mu}_{k+1}^{(\iota_k)} + \hat{\mu}_{k+1}^{(\iota_k+1)}}{1 + \frac{\hat{\Sigma}_{k+1}^{(\iota_k-1)}}{\delta^2} + 4\frac{\hat{\Sigma}_{k+1}^{(\iota_k)}}{\delta^2} + \frac{\hat{\Sigma}_{k+1}^{(\iota_k+1)}}{\delta^2}} \begin{bmatrix}
				\frac{\hat{\Sigma}_{k+1}^{(\iota_k-1)}}{\delta^2} \\ -2\frac{\hat{\Sigma}_{k+1}^{(\iota_k)}}{\delta^2} \\ \frac{\hat{\Sigma}_{k+1}^{(\iota_k+1)}}{\delta^2}
			\end{bmatrix}
			\end{align} 
			\begin{align} \nonumber
			&=  - \frac{\frac{\mathbf{c}^T \mathbf{A} \boldsymbol\xi_k^{(\iota_k-1)}}{\mathbf{c}^T \mathbf{A} \boldsymbol\phi_k^{(\iota_k-1)}} - 2 \frac{\mathbf{c}^T \mathbf{A} \boldsymbol\xi_k^{(\iota_k)}}{\mathbf{c}^T \mathbf{A} \boldsymbol\phi_k^{(\iota_k)}} + \frac{\mathbf{c}^T \mathbf{A} \boldsymbol\xi_k^{(\iota_k+1)}}{\mathbf{c}^T \mathbf{A} \boldsymbol\phi_k^{(\iota_k+1)}}}{1 + \frac{1}{\nu^2\mathbf{c}^T \mathbf{A} \boldsymbol\phi_k^{(\iota_k-1)}} + 4\frac{1}{\nu^2\mathbf{c}^T \mathbf{A} \boldsymbol\phi_k^{(\iota_k)}} + \frac{1}{\nu^2\mathbf{c}^T \mathbf{A} \boldsymbol\phi_k^{(\iota_k+1)}}} \\ \nonumber & \qquad \times \begin{bmatrix}
				\frac{1}{\nu^2\mathbf{c}^T \mathbf{A} \boldsymbol\phi_k^{(\iota_k-1)}} \\ -2\frac{1}{\nu^2\mathbf{c}^T \mathbf{A} \boldsymbol\phi_k^{(\iota_k)}} \\ \frac{1}{\nu^2 \mathbf{c}^T \mathbf{A} \boldsymbol\phi_k^{(\iota_k+1)}}
			\end{bmatrix} + \begin{bmatrix}
			\frac{\mathbf{c}^T \mathbf{A} \boldsymbol\xi_k^{(\iota_k-1)}}{\mathbf{c}^T \mathbf{A} \boldsymbol\phi_k^{(\iota_k-1)}} \\ \frac{\mathbf{c}^T \mathbf{A} \boldsymbol\xi_k^{(\iota_k)}}{\mathbf{c}^T \mathbf{A} \boldsymbol\phi_k^{(\iota_k)}} \\ \frac{\mathbf{c}^T \mathbf{A} \boldsymbol\xi_k^{(\iota_k+1)}}{\mathbf{c}^T \mathbf{A} \boldsymbol\phi_k^{(\iota_k+1)}}
			\end{bmatrix} = \begin{bmatrix}
				\frac{\theta_{k}^{(\iota_k-1)}}{\zeta_k^{(\iota_k)}} \\ \frac{\theta_{k}^{(\iota_k)}}{\zeta_k^{(\iota_k)}} \\ \frac{\theta_{k}^{(\iota_k+1)}}{\zeta_k^{(\iota_k)}}
			\end{bmatrix},
	\end{align}}
	with $\theta_{k}^{(\iota_k-1)}$, $\theta_{k}^{(\iota_k)}$, $\theta_{k}^{(\iota_k+1)}$, $\zeta_k^{(\iota_k)}$ and $\nu$ according to Lemma \ref{lem:model_solution}.
\end{proof}

\section{Proof of Lemma \ref{lem:maximizer}} \label{app:proof_lem_maximizer}

\begin{proof}[of Lemma \ref{lem:maximizer}]
To prove the lemma, we show that, for any $N \in \mathbb{N}$, the maximizer $u_{k+N}^*$ cannot be too far from the previous maximizer $u_k^*$ (such that the bound in \eqref{eq:lem_maximizer} is satisfied) under the given assumptions. We apply the following reasoning. For any $i \geq \iota_k$, it follows from Assumptions~\ref{assum:maximum} and \ref{assum:curvature}, and from $u_k^* = \iota_k^* \Delta_u$ that 
\begin{align}
		&f_k(u_k^*) - f_k(u^{(i)}) = \sum_{j = \iota_k^*}^{i-1} \left( f_k(u^{(j)}) - f_k(u^{(j+1)}) \right)\\ \nonumber
		& \geq L_b \sum_{j = \iota_k^*}^{i-1} \left( u^{(j+\frac{1}{2})} - u_k^* \right) = L_b \Delta_u \sum_{j = 0}^{i-\iota_k^*-1} \left( j + \frac{1}{2} \right) \\ \nonumber
		&= \frac{L_b \Delta_u}{2} (i-\iota_k^*)^2 .
\end{align}
Similarly, for any $i \leq \iota_k$, we can use the same arguments to show that
\begin{align}
		&f_k(u_k^*) - f_k(u^{(i)}) \geq \frac{L_b \Delta_u}{2} (\iota_k^*-i)^2 .
\end{align}
Thus, for all $i \in \mathbb{Z}$, we have that
\begin{equation} \label{eq:proof_Lb}
	f_k(u_k^*) - f_k(u^{(i)}) \geq \frac{L_b \Delta_u}{2} |i-\iota_k^*|^2 .
\end{equation}
By applying the triangle inequality and the inequality in Assumption~\ref{assum:time_change}, it follows that
\begin{align} \nonumber
	|f_{k+N}(u^{(i)}) - f_{k}(u^{(i)})| &= \left|\sum_{j = 0}^{N-1} \left( f_{k+j+1}(u^{(i)}) - f_{k+j}(u^{(i)}) \right) \right| \\
	& \leq N L_k,
\end{align}
for all $i \in \mathbb{Z}$ and all $N \in \mathbb{N}$, which implies that
\begin{equation} \label{eq:proof_Lk}
	f_{k}(u^{(i)}) - N L_k \leq f_{k+N}(u^{(i)}) \leq f_{k}(u^{(i)}) + N L_k,
\end{equation}
for all $i \in \mathbb{Z}$  and all $N \in \mathbb{N}$. Combining the inequalities in \eqref{eq:proof_Lb} and \eqref{eq:proof_Lk} yields
\begin{equation} \label{eq:proof_combined}
	f_{k+N}(u_k^*) - f_{k+N}(u^{(i)}) \geq \frac{L_b \Delta_u}{2} |i-\iota_k^*|^2 - 2N L_k,
\end{equation}
for all $i \in \mathbb{Z}$ and all $N \in \mathbb{N}$. Now, if the right-hand side of \eqref{eq:proof_combined} is larger than zero for some $i \in \mathbb{Z}$ and some $N \in \mathbb{N}$, then we have that $f_{k+N}(u^{(i)}) < f_{k+N}(u_k^*) \leq \max_{u \in \mathcal{U}} f_{k+N}(u)$ and, thus, $u^{(i)} \neq u_{k+N}^*$ by definition; see Assumption~\ref{assum:maximum}. The right-hand side of \eqref{eq:proof_combined} is larger than zero if and only if
\begin{equation}
	|i-\iota_k^*| > 2 \sqrt{\frac{L_k}{L_b \Delta_u}N}.
\end{equation}
It follows that $u^{(i)} \neq u_{k+N}^*$ if $|i-\iota_k^*| > 2\sqrt{\frac{L_k}{L_b \Delta_u}N}$, or, equivalently, $u^{(i)} \neq u_{k+N}^*$ if
\begin{equation}
	\left| u^{(i)} - u_k^* \right| > 2 \sqrt{\frac{L_k \Delta_u}{L_b}N}.
\end{equation}
Hence, all points $u^{(i)}$ that are too far removed from $u_k^*$ (i.e., for which $|i-\iota^*_k|$ is too large) are eliminated as possible locations of $u_{k+N}^*$. Then, we obtain that
\begin{equation} \label{eq:proof_step1}
	\left| u_{k+N}^* - u_k^* \right| \leq 2 \sqrt{\frac{L_k \Delta_u}{L_b}N},
\end{equation}
for all $N \in \mathbb{N}$, which implies that \eqref{eq:lem_maximizer} holds with $L^* = 2 \sqrt{\frac{L_k \Delta_u}{L_b}}$.
\end{proof}

\section{Proof of Lemma \ref{lem:convergence_bounds}} \label{app:proof_lem_convergence_bounds}

\begin{proof}[of Lemma \ref{lem:convergence_bounds}]
The proof consists of two steps. First, we show that for a sufficiently large value of $N \in \mathbb{N}$, we have that $|u_k - u_k^*| \geq \gamma$ implies that $|u_{k+N} - u_{k+N}^*| \leq |u_k - u_k^*| - \Delta_u$ for any $k \geq k_0$ and some positive constant $\gamma \in \mathbb{R}_{>0}$ under the conditions of the lemma. Second, we show that this means that the inequalities in \eqref{eq:lem_sup} and \eqref{eq:lem_lim_sup} hold for some positive constants $c_1$ and $c_2$.

\textit{Step 1}:
We obtain from Lemma~\ref{lem:maximizer} that
\begin{equation} \label{eq:proof_bound1}
	|u_{k+i}^* - u_k^*| \leq  L^* \sqrt{N},
\end{equation}
for all $i \in \{0,1,\dots,N\}$. Moreover, it follows from the condition in the lemma that $|u_{k+1} - u_k| \leq \Delta_u$ for all $k \in \mathbb{N}$ that
\begin{equation} \label{eq:proof_bound2}
	|u_{k+i} - u_k| \leq (N - 1) \Delta_u,
\end{equation}
for all $i \in \{0,1,\dots,N-1\}$. Let $\gamma =  L^* \sqrt{N} + (N-1) \Delta + b$. By combining the bounds in \eqref{eq:proof_bound1} and \eqref{eq:proof_bound2}, we obtain that $|u_k - u_{k}^*| \geq \gamma$ implies that
\begin{equation}
	|u_{k+i} - u_{k+i+1}^*| \geq b,
\end{equation}
for all $i \in \{0,1,\dots,N-1\}$. From the conditions in the lemma, it follows that $|u_k - u_k^*| \geq \gamma$ for some $k \geq k_0$ yields
\begin{equation} \label{eq:proof_bound3}
	|u_{k+i+1} - u_{k+i+1}^*| = |u_{k+i} - u_{k+i+1}^*| - \Delta_u,
\end{equation}
for all $i \in \{0,1,\dots,N-1\}$. Note that \eqref{eq:proof_bound3} implies that $u_{k+i+1} = u_{k+i} - \Delta_u$ if $u_{k+i} > u_{k+i+1}^*$ and $u_{k+i+1} = u_{k+i} + \Delta_u$ if $u_{k+i} < u_{k+i+1}^*$. From \eqref{eq:proof_bound1} and \eqref{eq:proof_bound2}, it follows that, if $|u_k - u_k^*| \geq \gamma$, we have that $u_{k+i} > u_{k}^*$ if $u_{k+i} > u_{k+i+1}^*$ and $u_{k+i} < u_{k}^*$ if $u_{k+i} < u_{k+i+1}^*$. Thus, we obtain that $u_{k+i+1} = u_{k+i} - \Delta_u$ if $u_{k+i} > u_{k}^*$ and $u_{k+i+1} = u_{k+i} + \Delta_u$ if $u_{k+i} < u_{k}^*$, which implies that
\begin{equation}
	|u_{k+i+1} - u_{k}^*| = |u_{k+i} - u_{k}^*| - \Delta_u,
\end{equation}
for all $i \in \{0,1,\dots,N-1\}$ and all $k \geq k_0$ if $|u_k - u_k^*| \geq \gamma$. This yields
\begin{equation} \label{eq:proof_bound4}
	|u_{k+N} - u_{k}^*| = |u_{k} - u_{k}^*| - N\Delta_u,
\end{equation}
for all $k \geq k_0$. Combining \eqref{eq:proof_bound1} and \eqref{eq:proof_bound4} gives
\begin{equation} \label{eq:proof_bound5}
	|u_{k+N} - u_{k+N}^*| \leq |u_{k} - u_{k}^*| + L^* \sqrt{N} - N\Delta_u,
\end{equation}
for any $k \geq k_0$ if $|u_k - u_k^*| \geq \gamma$. By letting $N = \left\lceil \frac{L^*}{\Delta_u} + 1 \right\rceil^2$, we obtain directly from \eqref{eq:proof_bound5} that
\begin{equation} \label{eq:proof_bound6}
	|u_{k+N} - u_{k+N}^*| \leq |u_{k} - u_{k}^*| -\Delta_u,
\end{equation}
if $|u_k - u_k^*| \geq \gamma$ for any $k \geq k_0$, which completes the first step of the proof.

\textit{Step 2}:
We have established that, for any $k \geq k_0$, we have that $|u_{k+N} - u_{k+N}^*| \leq |u_{k} - u_{k}^*| -\Delta_u$ whenever $|u_k - u_k^*| \geq \gamma$. Next, note that the condition $|u_{k+1} - u_k| \leq \Delta_u$ for all $k \in \mathbb{N}$ (see the description of the lemma) implies that
\begin{equation} \label{eq:proof_bound7}
	\max_{0\leq k \leq k_0+N} |u_{k} - u_{k}^*| \leq |u_{0} - u_{0}^*| + (k_0 + N) \Delta_u.
\end{equation}
Moreover, from the same condition and \eqref{eq:proof_bound1}, we have that $|u_{k+N} - u_{k+N}^*| \leq \gamma+L^* \sqrt{N}+N\Delta_u$ whenever $|u_k - u_k^*| \leq \gamma$. This, combined with the fact that $|u_{k+N} - u_{k+N}^*| \leq |u_{k} - u_{k}^*| -\Delta_u$ whenever $|u_k - u_k^*| \geq \gamma$ (see \eqref{eq:proof_step1}), results in the bound
\begin{equation} \label{eq:proof_bound8}
	|u_{k+N} - u_{k+N}^*| \leq \max\left\{|u_{k} - u_{k}^*| -\Delta_u, \gamma+L^* \sqrt{N}+N\Delta_u \right\},
\end{equation}
for all $k \geq k_0$. By combining the bounds in \eqref{eq:proof_bound7} and \eqref{eq:proof_bound8}, we obtain that
\begin{equation} \label{eq:proof_bound9}
	\sup_{k\geq 0} |u_{k} - u_{k}^*| \leq |u_{0} - u_{0}^*| + \gamma + L^* \sqrt{N} + (k_0 + N) \Delta_u.
\end{equation}
Moreover, it follows from \eqref{eq:proof_bound8} that
\begin{equation} \label{eq:proof_bound10}
	\limsup_{k \rightarrow \infty} |u_{k} - u_{k}^*| \leq \gamma + L^* \sqrt{N} + N\Delta_u.
\end{equation}
Hence, the bounds in \eqref{eq:lem_sup} and \eqref{eq:lem_lim_sup} of the lemma are satisfied with $c_1 = \gamma + L^* \sqrt{N} + (k_0 + N) \Delta_u$ and $c_2 = \gamma + L^* \sqrt{N} + N\Delta_u$, where $\gamma =  L^* \sqrt{N} + (N-1) \Delta + b$ and $N = \left\lceil \frac{L^*}{\Delta_u} + 1 \right\rceil^2$.  	
	
\end{proof} 

\section{Proof of Theorem \ref{thm:convergence}} \label{app:proof_thm}
\begin{proof}[of Theorem \ref{thm:convergence}]
	The aim is to show is that the conditions of Lemma~\ref{lem:convergence_bounds} are met. We show that there exists a constant $b \in \mathbb{R}_{>0}$, such that, whenever $|u_k - u_{k+1}^*| \geq b$, we have that $|u_{k+1} - u_{k+1}^*| = |u_k - u_{k+1}^*| - \Delta_u$ for any $k \geq 1$ under the conditions of the theorem.
	
	The proof consists of two steps. First, we show that there exists a positive constant $d \in \mathbb{R}_{>0}$ and two functions $\alpha_1, \alpha_2 \in \mathcal{K}$ (i.e., any function $\alpha: \mathbb{R}_{\geq 0} \rightarrow \mathbb{R}_{\geq 0}$ in the $\mathcal{K}$-class is continuous, zero at zero, and strictly increasing), such that, for all $\nu \in (0,\nu^*]$, all $\lambda^* \in (0,\lambda^{**}]$, all $\lambda \in (0,\lambda^*]$, any $M \in \{0,1\}$, and any $k \geq 1$, we have
	\begin{equation} \label{eq:proof_conv_bound1}
		\left|\hat{\mu}_{k+1}^{(i)} - f_{k+1}(u^{(i)}) \right| \leq d
	\end{equation}
	for all $i \in \{\iota_k, \iota_{k-1}\}$. In addition,
	\begin{equation} \label{eq:proof_conv_bound2}
		\left|h_{k+1}^{(i)} - l_{k+1}^{(i)}\right| \leq \alpha_1(\lambda^*) + \alpha_2(\lambda^*)|u_k - u_{k+1}^*|,
	\end{equation}    
	for all $i \in \{\iota_k-1, \iota_k, \iota_k+1 \}$. The bounds \eqref{eq:proof_conv_bound1} and \eqref{eq:proof_conv_bound2} hold under the assumption that we have continuously perturbed until that point in time (i.e., $|u_{r+1} - u_r| = \Delta_u$ for all $r \in \{0,1,\dots,k-1\}$). Here, $l_{k+1}^{(i)}$ is defined by
	\begin{equation} \label{eq:proof_conv_defn_l}
		l_{k+1}^{(i)} = \left\{ \begin{aligned}        
			& \hat{\mu}_{k+1}^{(i)}, & & \mbox{ if $i \in \{\iota_k, \iota_{k-1}\}$,} \\ 
			&  2\hat{\mu}_{k+1}^{(\iota_k)} - \hat{\mu}_{k+1}^{(\iota_{k-1})}, & & \mbox{otherwise,}        
		\end{aligned}\right.
	\end{equation}
	for all $i \in \{\iota_k-1, \iota_k, \iota_k+1 \}$, and $\lambda^{**} \in (0,1)$ is a constant. Second, we show that, for a sufficiently small value of $\lambda^* \in (0,\lambda^{**}]$, we always perturb and that the bounds in \eqref{eq:proof_conv_bound1} and \eqref{eq:proof_conv_bound2} imply that there exists a constant $b \in \mathbb{R}_{>0}$, such that, whenever $|u_k - u_{k+1}^*| \geq b$, we have that $|u_{k+1} - u_{k+1}^*| = |u_k - u_{k+1}^*| - \Delta_u$ for any $k \geq 1$ under the conditions of the theorem. Thus, the conditions of Lemma~\ref{lem:convergence_bounds} are met.
	
	\textit{Step 1}:
	We first show that the bound in \eqref{eq:proof_conv_bound1} holds for all $\lambda \in (0,\lambda^*]$, all $\lambda^*\in(0,\lambda^{**}]$ (which implies that \eqref{eq:proof_conv_bound1} holds for all $\lambda \in (0,\lambda^{**}]$ ), all $i \in \{\iota_k, \iota_{k-1}\}$, and any $M \in \{0,1\}$. Let us set $\lambda^{**} = 0.5$. Consider any $k \geq 1$ and assume that $|u_{r+1} - u_r| = \Delta_u$ for all $r \in \{0,1,\dots,k-1\}$. Under Assumption~\ref{assum:time_change}, we obtain from \eqref{eq:yk} that
	\begin{equation} \label{eq:proof_conv_measurement}
		|y_j - f_{k+1}(u_j)| \leq \rho + (k+1-j)L_k
	\end{equation}
	for all $j \in \{0,1,\dots,k\}$. By combining \eqref{eq:proof_conv_measurement} and the definition of the mean in \eqref{eq:mean}, we obtain
	\begin{align} \nonumber
		&\left| \hat{\mu}_{k+1}^{(i)} - f_{k+1}(u^{(i)}) \right| = \left| \frac{\sum_{j \in \mathcal{J}_k^{(i)}}\omega_{j,k+1} \left(y_j - f_{k+1}(u^{(i)}) \right)}{\sum_{j \in \mathcal{J}_k^{(i)}}\omega_{j,k+1}}  \right| \\ \nonumber
		\quad &\leq \rho + L_k \frac{\sum_{j \in \mathcal{J}_k^{(i)}} (k+1 - j)\omega_{j,k+1} }{\sum_{j \in \mathcal{J}_k^{(i)}}\omega_{j,k+1}} \\  \label{eq:proof_conv_mean_bound1}
		\quad & \leq \rho + L_k \frac{\sum_{j = -\infty}^{p_k^{(i)}} (k+1 - j)\omega_{j,k+1} }{\omega_{p_k^{(i)},k+1}}
	\end{align}
	for all $i$ for which $\mathcal{J}_k^{(i)} \neq \emptyset$; see the definition of $p_k^{(i)}$ in \eqref{eq:p_k} for the last inequality. For $M = 0$, we have that $\omega_{j,k} = \lambda^{k-j}$; see \eqref{eq:omega}. It follows from \eqref{eq:proof_conv_mean_bound1} that
	\begin{align} \nonumber
		\left| \hat{\mu}_{k+1}^{(i)} - f_{k+1}(u^{(i)}) \right| &\leq \rho + L_k \frac{\sum_{j = -\infty}^{p_k^{(i)}} (k+1 - j)\lambda^{k+1-j} }{\lambda^{k+1-p_{k}^{(i)}}}\\ \nonumber
		&= \rho + L_k \sum_{r = 0}^\infty (k+1 - p_k^{(i)} + r) \lambda^r\\ \nonumber
		&= \rho + L_k \left( \frac{k+1-p_k^{(i)}}{1-\lambda} + \frac{\lambda}{(1-\lambda)^2} \right)\\ \label{eq:proof_conv_mean_bound2}
		& \leq \rho + \left( 2(k+1-p_k^{(i)}) + 2\right) L_k
	\end{align}
	for all $\lambda \in (0,\lambda^{**}]$ and all $i$ for which $\mathcal{J}_k^{(i)} \neq \emptyset$, where we have used that $\lambda^{**} = 0.5$, $\sum_{r = 0}^\infty \lambda^r = \frac{1}{1-\lambda}$, $\sum_{r = 0}^\infty r \lambda^r = \frac{\lambda}{(1 - \lambda)^2}$, $p_k^{(\iota_k)} = k$, and $p_k^{(\iota_{k-1})} = {k-1}$. Similarly, for $M = 1$, we have that $\omega_{j,k} = \left(1 + \ln\left( \frac{1}{\lambda}(k-j) \right) \right)\lambda^{k-j}$, such that
	{\small{\begin{align} \nonumber
				&\left| \hat{\mu}_{k+1}^{(i)} - f_{k+1}(u^{(i)}) \right| \leq \\ \nonumber
				&\leq \rho + L_k \frac{\sum_{j = -\infty}^{p_k^{(i)}} \left(k+1 - j + \ln\left( \frac{1}{\lambda} \right) (k+1-j)^2\right)\lambda^{k+1-j} }{\left(1 + \ln\left( \frac{1}{\lambda} \right) (k+1-p_k^{(i)})\right)\lambda^{k+1-p_{k}^{(i)}}}\\ \nonumber
				&= \rho + \frac{L_k}{1 + \ln\left( \frac{1}{\lambda} \right) (k+1-p_k^{(i)})}\\ \nonumber
				&\times \sum_{r=0}^{\infty} \left(k+1 - p_k^{(i)} + r + \ln\left( \frac{1}{\lambda} \right) (k+1-p_k^{(i)} + r)^2\right)\lambda^{k+1-j} \\ \nonumber 			
				&= \rho + L_k \frac{ \left(k+1 - p_k^{(i)} + \ln\left( \frac{1}{\lambda} \right) (k+1-p_k^{(i)} )^2\right)\frac{1}{1-\lambda} }{1 + \ln\left( \frac{1}{\lambda} \right) (k+1-p_k^{(i)})}\\ \nonumber
				& \quad + L_k \frac{\left(1 + 2\ln\left( \frac{1}{\lambda} \right) (k+1-p_k^{(i)})\right) \frac{\lambda}{(1-\lambda)^2} + \ln\left( \frac{1}{\lambda} \right) \frac{\lambda(1+\lambda)}{(1-\lambda)^3}}{1 + \ln\left( \frac{1}{\lambda} \right) (k+1-p_k^{(i)})}\\ \nonumber
				&\leq \rho + L_k\left( \frac{2\left(k+1-p_k^{(i)}\right)}{1-\lambda} + \frac{3\lambda}{(1-\lambda)^2} + \ln\left( \frac{1}{\lambda} \right) \frac{\lambda(1+\lambda)}{(1-\lambda)^3} \right)\\ 
				&\leq \rho + \left( 4(k+1-p_k^{(i)}) + 12 \right) L_k \label{eq:proof_conv_mean_bound3}
	\end{align}}}
	for all $\lambda \in (0,\lambda^{**}]$ and all $i$ for which $\mathcal{J}_k^{(i)} \neq \emptyset$. Here, we have used that $\lambda^{**} = 0.5$, $\sum_{r = 0}^\infty \lambda^r = \frac{1}{1-\lambda}$, $\sum_{r = 0}^\infty r \lambda^r = \frac{\lambda}{(1 - \lambda)^2}$, $\sum_{r = 0}^\infty r^2 \lambda^r = \frac{\lambda(1+\lambda)}{(1 - \lambda)^3}$, and that $\ln\left( \frac{1}{\lambda} \right) \lambda \leq 0.5$ for all $\lambda \in (0,\lambda^*]$. Noting that $p_k^{(\iota_k)} = k$ and $p_k^{(\iota_{k-1})} = {k-1}$, we obtain from \eqref{eq:proof_conv_mean_bound2} and \eqref{eq:proof_conv_mean_bound3} that the bound in \eqref{eq:proof_conv_bound1} holds with $d = \rho + 20 L_k$ for all $i \in \{\iota_k,\iota_{k-1}\}$, all $\lambda \in (0,\lambda^{**}]$, and any $M \in \{0,1\}$.
	
	Next, we show that the bound in \eqref{eq:proof_conv_bound2} holds for $i \in \{\iota_k-1,\iota_k,\iota_k+1\}$, all $\nu \in (0,\nu^*]$, all $\lambda^* \in (0,\lambda^{**}]$ (with $\lambda^{**} = 0.5$), all $\lambda \in (0,\lambda^*]$, and any $M \in \{0,1\}$. Again, consider any $k \geq 1$ and assume that $|u_{r+1} - u_r| = \Delta_u$ for all $r \in \{0,1,\dots,k-1\}$. Note that, if $\mathcal{J}_k^{(\iota_k+1)} = \emptyset$ or if $\mathcal{J}_k^{(\iota_k-1)} = \emptyset$, it follows from \eqref{eq:h_three_points} and \eqref{eq:proof_conv_defn_l} that $h_{k+1}^{(i)} = l_{k+1}^{(i)}$ for all $i \in \{\iota_k-1,\iota_k,\iota_k+1\}$, which implies that the bound in \eqref{eq:proof_conv_bound2} is satisfied for any functions $\alpha_1,\alpha_2 \in \mathcal{K}$. What remains to be shown is that the bound \eqref{eq:proof_conv_bound2} holds for some functions $\alpha_1,\alpha_2 \in \mathcal{K}$ if all three points with indices $i \in \{\iota_k-1,\iota_k,\iota_k+1\}$ have been measured before (when the third condition in \eqref{eq:h_three_points} is valid). Note that $\iota_{k-1} \in\{ \iota_k-1,\iota_k+1 \}$ (i.e., the index of $u_{k-1} = u^{(\iota_{k-1})}$) because we have continuously perturbed until now. Now, let us denote index of the third point (that is not equal to $\iota_k$ or $\iota_{k-1}$) by $\iota_{k-q_k} \in\{ \iota_k-1,\iota_k+1 \}$. Note that $\iota_{k-q_k}$ is index of the point for which the most recent measurement is oldest. Because the input can only move one step on the input grid at each time step, we always have that $q_k \geq 3$, which implies that $p_k^{(\iota_{k-q_k})} \leq k - 3$. It follows from \eqref{eq:h_three_points} that
	\begin{align} \nonumber
		\left| h_{k+1}^{(\iota_k)} - l_{k+1}^{(\iota_k)} \right| =&  \frac{2 \hat{\Sigma}_{k+1}^{(\iota_k)}}{(\nu\rho)^2 + \hat{\Sigma}_{k+1}^{(\iota_{k-1})}   + 4 \hat{\Sigma}_{k+1}^{(\iota_k)} + \hat{\Sigma}_{k+1}^{(\iota_{k-q_k})}} \\ & \times  \left|\hat{\mu}_{k+1}^{(\iota_{k-1})} - 2 \hat{\mu}_{k+1}^{(\iota_k)}  + \hat{\mu}_{k+1}^{(\iota_{k-q_k})} \right|, \label{eq:proof_conv_h_bound1}
	\end{align}
	\begin{align} \nonumber
		\left| h_{k+1}^{(\iota_{k-1})} - l_{k+1}^{(\iota_{k-1})} \right|  = &  \frac{\hat{\Sigma}_{k+1}^{(\iota_{k-1})}}{(\nu\rho)^2 + \hat{\Sigma}_{k+1}^{(\iota_{k-1})}  + 4 \hat{\Sigma}_{k+1}^{(\iota_k)} + \hat{\Sigma}_{k+1}^{(\iota_{k-q_k})}} \\ & \times \left|\hat{\mu}_{k+1}^{(\iota_{k-1})} - 2 \hat{\mu}_{k+1}^{(\iota_k)}  + \hat{\mu}_{k+1}^{(\iota_{k-q_k})} \right| \label{eq:proof_conv_h_bound2}
	\end{align}
	and
	\begin{align} \label{eq:proof_conv_h_bound3}
		&\left| h_{k+1}^{(\iota_{k-q_k})} - l_{k+1}^{(\iota_{k-q_l})} \right| = \left|\hat{\mu}_{k+1}^{(\iota_k-1)} - 2 \hat{\mu}_{k+1}^{(\iota_k)}  + \hat{\mu}_{k+1}^{(\iota_k+1)} \right| \\ \nonumber 
		&\quad \times   \left( 1 - \frac{ \hat{\Sigma}_{k+1}^{(\iota_{k-q_k})}}{(\nu\rho)^2 + \hat{\Sigma}_{k+1}^{(\iota_{k-1})}  + 4 \hat{\Sigma}_{k+1}^{(\iota_k)} + \hat{\Sigma}_{k+1}^{(\iota_{k-q_k})}}
		\right)  \\ \nonumber
		& = \left|\hat{\mu}_{k+1}^{(\iota_{k-1})} - 2 \hat{\mu}_{k+1}^{(\iota_k)}  + \hat{\mu}_{k+1}^{(\iota_{k-q_k})} \right| \\ \nonumber & \quad \times \frac{ (\nu\rho)^2  + 4 \hat{\Sigma}_{k+1}^{(\iota_k)} + \hat{\Sigma}_{k+1}^{(\iota_{k-1})} }{(\nu\rho)^2 + \hat{\Sigma}_{k+1}^{(\iota_{k-1})}  + 4 \hat{\Sigma}_{k+1}^{(\iota_k)} + \hat{\Sigma}_{k+1}^{(\iota_{k-q_k})}} .
	\end{align}
	We obtain from the equations in \eqref{eq:proof_conv_h_bound1}-\eqref{eq:proof_conv_h_bound3} that
	\begin{align} \label{eq:proof_conv_h_bound4}
		&\left| h_{k+1}^{(i)} - l_{k+1}^{(i)} \right| \leq \frac{ (\nu\rho)^2  + 4 \hat{\Sigma}_{k+1}^{(\iota_k)} + \hat{\Sigma}_{k+1}^{(\iota_{k-1})} }{(\nu\rho)^2 + \hat{\Sigma}_{k+1}^{(\iota_k-1)}  + 4 \hat{\Sigma}_{k+1}^{(\iota_k)} + \hat{\Sigma}_{k+1}^{(\iota_k+1)}} \\ \nonumber & 
		\quad\times \left|\hat{\mu}_{k+1}^{(\iota_{k-1})} - 2 \hat{\mu}_{k+1}^{(\iota_k)}  + \hat{\mu}_{k+1}^{(\iota_{k-q_k})} \right| \\ \nonumber
		&\leq \frac{ (\nu\rho)^2  + 4 \hat{\Sigma}_{k+1}^{(\iota_k)} + \hat{\Sigma}_{k+1}^{(\iota_{k-1})} }{\hat{\Sigma}_{k+1}^{(\iota_{k-q_k})}} \left|\hat{\mu}_{k+1}^{(\iota_{k-1})} - 2 \hat{\mu}_{k+1}^{(\iota_k)}  + \hat{\mu}_{k+1}^{(\iota_{k-q_k})} \right|
	\end{align}
	for all $i \in \{ \iota_k-1, \iota_k,\iota_k+1 \}$ (or, alternatively, for all $i \in \{ \iota_k, \iota_{k-1},\iota_{k-q_k} \}$). To derive an upper bound on the right-hand side of the last inequality in \eqref{eq:proof_conv_h_bound4}, we need to derive upper and lower bounds on the variance-like variables $\hat{\Sigma}_{k+1}^{(i)}$ in \eqref{eq:variance} for $i \in \{ \iota_k, \iota_{k-1},\iota_{k-q_k} \}$. For $M = 0$, we have that $\omega_{j,k} = \lambda^{k-j}$; see \eqref{eq:omega}. It follows from the definition of $\hat{\Sigma}_{k+1}^{(i)}$ in \eqref{eq:variance} that 
	\begin{equation} \label{eq:proof_conv_h_bound5}
		\hat{\Sigma}_{k+1}^{(i)} \geq \frac{\rho^2}{\sum_{j = -\infty}^{p_k^{(i)}} \lambda^{k+1-j}} = \frac{\rho^2}{\lambda^{k+1-p_{k}^{(i)}} \sum_{r = 0}^\infty \lambda^r} = \frac{\rho^2(1 - \lambda)}{\lambda^{k+1-p_{k}^{(i)}}}
	\end{equation}
	and
	\begin{equation} \label{eq:proof_conv_h_bound6}
		\hat{\Sigma}_{k+1}^{(i)} \leq \frac{\rho^2}{\lambda^{k+1-p_{k}^{(i)}}};
	\end{equation}
	see \eqref{eq:p_k} for the definition of $p_{k}^{(i)}$. Similarly, for $M = 1$, we have from \eqref{eq:omega} that $\omega_{j,k} = \left( 1 + \ln\left( \frac{1}{\lambda} \right)(k-j)\right)\lambda^{k-j}$, in which case the following bounds on $\hat{\Sigma}_{k+1}^{(i)}$ can be derived:
	\begin{align} \label{eq:proof_conv_h_bound7}
		& \hat{\Sigma}_{k+1}^{(i)} \geq \frac{\rho^2}{\sum_{j = -\infty}^{p_k^{(i)}} \left( 1 + \ln\left( \frac{1}{\lambda} \right)(k+1-j) \right) \lambda^{k+1-j}}\\ \nonumber
		&= \frac{\rho^2}{ \lambda^{k+1-p_{k}^{(i)}} \sum_{r = 0}^{\infty} \left( 1 + \ln\left( \frac{1}{\lambda} \right)(k+1-p_k^{(i)}+r) \right) \lambda^{r}}\\ \nonumber
		&=\frac{\rho^2(1-\lambda)^2}{ \lambda^{k+1-p_{k}^{(i)}} \left( \left( 1 + \ln\left( \frac{1}{\lambda} \right)(k+1-p_k^{(i)}) \right) (1-\lambda)  + \ln\left( \frac{1}{\lambda} \right) \lambda \right)} \\ \nonumber
		& \geq \frac{\rho^2(1-\lambda)^2}{ \lambda^{k+1-p_{k}^{(i)}} \left( 1 + \ln\left( \frac{1}{\lambda} \right)(k+1-p_k^{(i)})\right)} 
	\end{align}
	and
	\begin{equation} \label{eq:proof_conv_h_bound8}
		\hat{\Sigma}_{k+1}^{(i)} \leq \frac{\rho^2}{\lambda^{k+1-p_{k}^{(i)}}\left( 1 + \ln\left( \frac{1}{\lambda} \right)(k+1-p_k^{(i)})\right)}.
	\end{equation}
	With the upper and lower bounds in \eqref{eq:proof_conv_h_bound5} and \eqref{eq:proof_conv_h_bound6}, we obtain
	\begin{align} \nonumber
		&	\frac{ (\nu\rho)^2  + 4 \hat{\Sigma}_{k+1}^{(\iota_k)} + \hat{\Sigma}_{k+1}^{(\iota_{k-1})} }{\hat{\Sigma}_{k+1}^{(\iota_{k-q_k})}} \leq \frac{\lambda^{k+1-p_{k}^{(\iota_{k-q_k})}}}{1-\lambda} \left( \nu^2 + \frac{4}{\lambda} + \frac{1}{\lambda^2} \right) \\
		&\quad \leq \frac{\lambda}{1-\lambda} \lambda^{k-3-p_{k}^{(\iota_{k-q_k})}}\left( \nu^2 + 5 \right) \label{eq:proof_conv_h_bound9}
	\end{align}
	if $M = 0$. In addition, it follows from \eqref{eq:proof_conv_h_bound7} and \eqref{eq:proof_conv_h_bound8} that
	\begin{align} \label{eq:proof_conv_h_bound10}
		&\frac{ (\nu\rho)^2  + 4 \hat{\Sigma}_{k+1}^{(\iota_k)} + \hat{\Sigma}_{k+1}^{(\iota_{k-1})} }{\hat{\Sigma}_{k+1}^{(\iota_{k-q_k})}} \leq \\ \nonumber
		&\frac{\lambda^{k+1-p_{k}^{(\iota_{k-q_k})}}\left( 1 + \ln\left( \frac{1}{\lambda} \right)(k+1-p_k^{(\iota_{k-q_k})})\right)}{(1-\lambda)^2}  \\ \nonumber & \times \left( \nu^2 + \frac{4}{\lambda\left(1 + \ln\left( \frac{1}{\lambda} \right) \right)} + \frac{1}{\lambda^2 \left(1 + 2 \ln\left( \frac{1}{\lambda} \right) \right)} \right) \\ \nonumber
		& \leq  \frac{\lambda}{(1-\lambda)^2} \lambda^{k-3-p_{k}^{(\iota_{k-q_k})}}\left( 3 + \frac{1}{2}(k-3-p_k^{(\iota_{k-q_k})})\right) \left( \nu^2 + 5 \right)
	\end{align}
	if $M = 1$. Here, we have used that $\lambda \ln\left(\frac{1}{\lambda} \right) \leq \frac{1}{2}$ for all $\lambda \in (0,1)$. Note that the bound in \eqref{eq:proof_conv_h_bound10} is larger than the one in \eqref{eq:proof_conv_h_bound9}. Therefore, it is valid for both $M = 0$ and $M = 1$. Using Assumptions~\ref{assum:curvature} and \ref{assum:time_change}, we obtain the following bound:
	\begin{align} \nonumber 
		&\left|\hat{\mu}_{k+1}^{(\iota_{k-1})} - 2 \hat{\mu}_{k+1}^{(\iota_k)}  + \hat{\mu}_{k+1}^{(\iota_{k-q_k})} \right|  \leq \left|f_{k+1}(u^{(\iota_{k -1})}) - f_{k+1}(u^{(\iota_k)}) \right| \\ \nonumber
		& \quad + \left| f_{k+1}(u^{(\iota_{k-q_k})})  - f_{k+1}(u^{(\iota_k)}) \right| + \left| \hat{\mu}_{k+1}^{(\iota_{k-1})} - f_{k+1}(u^{(\iota_{k -1})}) \right| \\ \nonumber 
		& \quad + 2 \left| \hat{\mu}_{k+1}^{(\iota_k)} - f_{k+1}(u^{(\iota_k )}) \right| + \left| \hat{\mu}_{k+1}^{(\iota_{k-q_k})} - f_{k+1}(u^{(\iota_{k-q_k})}) \right| \\ \nonumber
		& \leq L_b \left|u^{(\iota_k-\frac{1}{2})} - u_{k+1}^* \right| + L_b \left|u^{(\iota_k+\frac{1}{2})} - u_{k+1}^* \right| + 4\rho \\ \nonumber 
		&\quad + 4\left( 2\left(k+1-p_k^{(\iota_{k-q_k})}\right) + 2\right) L_k \\ \nonumber
		& \leq 4\rho + 40 L_k + L_b \Delta_u + 8 \left(k-3-p_k^{(\iota_{k-q_k})}\right) L_k \\  
		& \quad + 2 L_b \left|u_k - u_{k+1}^* \right|. \label{eq:proof_conv_h_bound11}
	\end{align}
	Now, combining \eqref{eq:proof_conv_h_bound4}, \eqref{eq:proof_conv_h_bound9}, and \eqref{eq:proof_conv_h_bound11} yields
	{ \begin{align} \label{eq:proof_conv_h_boundf}
		&\left| h_{k+1}^{(i)} - l_{k+1}^{(i)} \right| \leq \frac{\lambda}{(1-\lambda)^2} \lambda^{k-3-p_{k}^{(\iota_{k-q_k})}} \\ \nonumber 
		& \quad\times \left( 3 + \frac{1}{2}(k-3-p_k^{(\iota_{k-q_k})})\right) \left( \nu^2 + 5 \right) \\ \nonumber
		& \quad\times \left( 4\rho + 40 L_k + L_b \Delta_u + 8 \left(k-3-p_k^{(\iota_{k-q_k})}\right) L_k \right. \\ \nonumber & \quad \left. + 2 L_b \left|u_k - u_{k+1}^* \right| \right) \\ \nonumber
		&= \frac{\lambda}{(1-\lambda)^2} \left( \nu^2 + 5 \right) \lambda^{k-3-p_{k}^{(\iota_{k-q_k})}} \\ \nonumber 
		&\quad\biggl( 12\rho + 120 L_k + 3L_b \Delta_u + \left( 44 L_k + 2\rho + \frac{L_b \Delta_u}{2} \right) \\ \nonumber 
		& \quad\left(k-3-p_k^{(\iota_{k-q_k})}\right) + 4 L_k\left(k-3-p_k^{(\iota_{k-q_k})}\right)^2 \\ \nonumber
		& \quad  + \left( 6 + k-3-p_k^{(\iota_{k-q_k})}\right) L_b \left|u_k - u_{k+1}^* \right| \biggr)\\ \nonumber
		&\leq \frac{\lambda}{(1-\lambda)^2} \left( \nu^2 + 5 \right) \Biggl( 12\rho + 120 L_k + 3L_b \Delta_u + \\ \nonumber 
		& \quad\left( 44 L_k + 2\rho + \frac{L_b \Delta_u}{2} \right) \frac{1}{\ln\left( \frac{1}{\lambda}\right) e} \\ \nonumber
		& \quad+ \frac{16 L_k}{\left( \ln\left( \frac{1}{\lambda}\right) e \right)^2}  + \left( 6 + \frac{1}{\ln\left( \frac{1}{\lambda}\right) e}\right) L_b \left|u_k - u_{k+1}^* \right| \Biggr) \\ \nonumber
		&\leq \frac{\lambda^*}{(1-\lambda^*)^2} \left( \left(\nu^*\right)^2 + 5 \right) \Biggl( 12\rho + 120 L_k + 3L_b \Delta_u \\ \nonumber 
		&\quad+ \left( 44 L_k + 2\rho + \frac{L_b \Delta_u}{2} \right) \frac{1}{\ln\left( \frac{1}{\lambda^*}\right) e} \\ \nonumber
		& \quad + \frac{16 L_k}{\left( \ln\left( \frac{1}{\lambda^*}\right) e \right)^2}  + \left( 6 + \frac{1}{\ln\left( \frac{1}{\lambda^*}\right) e}\right) L_b \left|u_k - u_{k+1}^* \right| \Biggr)
	\end{align}}
	for all $i \in \{ \iota_k-1, \iota_k,\iota_k+1 \}$, all $\nu \in (0,\nu^*]$, all $\lambda^* \in (0,\lambda^{**}]$, all $\lambda \in (0,\lambda^*]$, and any $M \in \{0,1\}$, where we have used that $r \lambda^r \leq \frac{1}{\ln\left( \frac{1}{\lambda}\right) e}$ and $r^2 \lambda^r \leq \frac{4}{\left( \ln\left( \frac{1}{\lambda}\right) e \right)^2}$ for all $r \in \mathbb{N}$. Note that the limits of the functions $\frac{\lambda^*}{(1-\lambda^*)^2}$ and $\frac{1}{\ln\left( \frac{1}{\lambda^*}\right)}$ in \eqref{eq:proof_conv_h_boundf} are zero as $\lambda^*$ approaches zero. Moreover, these functions are strictly increasing for $\lambda^* \in (0,\lambda^{**}]$. It follows from \eqref{eq:proof_conv_h_boundf} that the bound in \eqref{eq:proof_conv_bound2} holds for some function $\alpha_1,\alpha_2 \in \mathcal{K}$. This completes Step 1 of the proof.
	
	\textit{Step 2}: We first show that the conditions of Lemma~\ref{lem:convergence_bounds} hold for a sufficiently small value of $\lambda^*$ using the results of Step 1 under the assumption that we have perturbed until that point in time. Consider any $k \geq 1$ and assume that $|u_{r+1} - u_r| = \Delta_u$ for all $r \in \{0,1,\dots,k-1\}$. Note that either $u_{k+1}^* \leq u_k < u_{k-1}$, or $u_{k+1}^* \leq u_{k-1} < u_{k}$, or $u_k < u_{k-1} \leq u_{k+1}^*$, or $u_{k-1} < u_{k} \leq u_{k+1}^*$ holds. If $u_{k+1}^* \leq u_k < u_{k-1}$, then we obtain from Assumption~\ref{assum:curvature}, the bounds in \eqref{eq:proof_conv_bound1} and \eqref{eq:proof_conv_bound2}, and the definition in \eqref{eq:proof_conv_defn_l} that
	\begin{align} 
		& h_{k+1}^{(\iota_k-1)} - h_{k+1}^{(\iota_k)} \geq l_{k+1}^{(\iota_k - 1)} - l_{k+1}^{(\iota_k)}  \\ \nonumber 
		&\quad - 2\left( \alpha_1(\lambda^*) + \alpha_2(\lambda^*)|u_k - u_{k+1}^*| \right)\\ \nonumber
		&= \mu_{k+1}^{(\iota_k)} - \mu_{k+1}^{(\iota_{k-1})} - 2\left( \alpha_1(\lambda^*) + \alpha_2(\lambda^*)|u_k - u_{k+1}^*| \right) \\ \nonumber
		&\geq f_{k+1}(u_k) - f_{k+1}(u_{k-1}) - 2\left( \alpha_1(\lambda^*) \right. \\ \nonumber 
		& \quad \left. + \alpha_2(\lambda^*)|u_k - u_{k+1}^*| \right) - 2d \\ \nonumber
		&\geq L_b \left|u^{\left(\frac{\iota_k + \iota_{k-1}}{2}\right)} - u_{k+1}^*\right| \\ \nonumber 
		& \quad - 2\left( \alpha_1(\lambda^*) + \alpha_2(\lambda^*)|u_k - u_{k+1}^*| \right) - 2d \\ \nonumber
		&\geq (L_b - 2 \alpha_2(\lambda^*) ) \left|u_k - u_{k+1}^*\right| - 2 \alpha_1(\lambda^*) - 2d - \frac{L_b \Delta_u}{2},
	\end{align}
	and, using the same arguments, that
	\begin{align}
		&h_{k+1}^{(\iota_k)} - h_{k+1}^{(\iota_k+1)} \\ \nonumber 
		&\: \geq l_{k+1}^{(\iota_k)} - l_{k+1}^{(\iota_k+1)} - 2\left( \alpha_1(\lambda^*) + \alpha_2(\lambda^*)|u_k - u_{k+1}^*| \right)\\ \nonumber 
		&\: = \mu_{k+1}^{(\iota_k)} - \mu_{k+1}^{(\iota_{k-1})} - 2\left( \alpha_1(\lambda^*) + \alpha_2(\lambda^*)|u_k - u_{k+1}^*| \right) \\ \nonumber 
		&\: \geq (L_b - 2 \alpha_2(\lambda^*) ) \left|u_k - u_{k+1}^*\right| - 2 \alpha_1(\lambda^*) - 2d - \frac{L_b \Delta_u}{2}.
	\end{align}
	Assuming that $\lambda^*$ is sufficiently small such that $\lambda^* \leq \min\left\{\alpha_2^{-1}\left(\frac{L_b}{4}\right), \lambda^{**} \right\}$, we have that $h_{k+1}^{(\iota_k-1)} \geq h_{k+1}^{(\iota_k)} + 2\tau \geq h_{k+1}^{(\iota_k+1)} + 4 \tau$ whenever $\left|u_k - u_{k+1}^*\right| \geq b$, with $b = \frac{4}{L_b} \left( \alpha_1(\lambda^{**}) + d + \tau \right) + \Delta_u $. Subsequently, it follows from \eqref{eq:optimization1}-\eqref{eq:optimization2} that $u_{k+1} = u^{(\iota_k - 1)}$ whenever $|u_k - u_{k+1}^*| \geq b$, which implies that  $|u_{k+1} - u_{k+1}^*| = |u_k - u_{k+1}^*| - \Delta_u$ whenever $|u_k - u_{k+1}^*| \geq b$. A similar reasoning can be applied to show that $|u_{k+1} - u_{k+1}^*| = |u_k - u_{k+1}^*| - \Delta_u$ whenever $|u_k - u_{k+1}^*| \geq b$ for any of the three remaining conditions $u_{k+1}^* \leq u_{k-1} < u_{k}$, $u_k < u_{k-1} \leq u_{k+1}^*$, and $u_{k-1} < u_{k} \leq u_{k+1}^*$, assuming $\lambda^* \leq \min\left\{\alpha_2^{-1}\left(\frac{L_b}{4}\right), \lambda^{**} \right\}$. Hence, the conditions of Lemma~\ref{lem:convergence_bounds} are met for a sufficiently small value of $\lambda^*$ if we always perturb. 
	
	We show next that we always perturb for a sufficiently small value of $\lambda^*$ under the conditions of the theorem. Note that not perturbing implies that $u_{k+1} = u_k$. We show by contradiction that there does not exist a $k \in \mathbb{N}$ for which $u_{k+1} = u_k$ if $\lambda^*$ is sufficiently small. Let $k_1 \in \mathbb{N}$ be the first time instance for which $u_{k_1+1} = u_{k_1}$, if it exists. Obviously, we cannot have that $k_1 = 0$ because $u_1 \in \{u_0 - \Delta_u, u_0 + \Delta_u\}$. Hence, we must have that $k_1 \geq 1$. It follows from \eqref{eq:optimization1}-\eqref{eq:optimization2} that $u_{k_1+1} = u_{k_1}$ only if $h_{k_1+1}^{(\iota_{k_1})} \geq h_{k_1+1}^{(\iota_{k_1}-1)}$, $h_{k_1+1}^{(\iota_{k_1})} \geq h_{k_1+1}^{(\iota_{k_1}+1)}$, and $h_{k_1+1}^{(\iota_{k_1})} \geq h_{k_1+1}^{(\iota_{{k_1}-1})} + \tau$, where we note that $\iota_{{k_1}-1} \in \{\iota_{k_1}-1,\iota_{k_1}+1\}$. This implies that the condition
	\begin{equation} \label{eq:proof_conv_h_condition}
		h_{k_1+1}^{(\iota_{k_1}-1)} - 2h_{k_1+1}^{(\iota_{k_1})} + h_{k_1+1}^{(\iota_{k_1}+1)} \leq -\tau
	\end{equation}
	must hold if $k_1 \geq 1$ is the first time instance such that $u_{k_1+1} = u_{k_1}$. From the bound in \eqref{eq:proof_conv_bound2} and the definition in \eqref{eq:proof_conv_defn_l}, we obtain that
	\begin{align} \nonumber
		&\left|h_{k_1+1}^{(\iota_{k_1}-1)} - 2h_{k_1+1}^{(\iota_{k_1})} + h_{k_1+1}^{(\iota_{k_1}+1)}\right| \leq \left|l_{k_1+1}^{(\iota_{k_1}-1)} - 2l_{k_1+1}^{(\iota_{k_1})} \right. \\ \nonumber 
		& \quad \left. + l_{k_1+1}^{(\iota_{k_1}+1)}\right| + 4 \alpha_1(\lambda^*) + 4\alpha_2(\lambda^*)|u_k - u_{k+1}^*| \\
		& = 4 \alpha_1(\lambda^*) + 4\alpha_2(\lambda^*)|u_k - u_{k+1}^*|. \label{eq:proof_conv_h_bound}
	\end{align}
	Moreover, we have shown in the beginning of Step 2 that $|u_{k+1} - u_{k+1}^*| = |u_k - u_{k+1}^*| - \Delta_u$ whenever $|u_k - u_{k+1}^*| \geq b$ under the given conditions, which means that $k_1$ cannot be the first time instance for which $u_{k_1+1} = u_{k_1}$ if $|u_{k_1} - u_{k_1+1}^*| \geq b$. Thus, we must have that $|u_{k_1} - u_{k_1+1}^*| < b$. Therefore, it follows using \eqref{eq:proof_conv_h_bound} that
	\begin{equation} \label{eq:proof_conv_h_bound12}
		\left|h_{k_1+1}^{(\iota_{k_1}-1)} - 2h_{k_1+1}^{(\iota_{k_1})} + h_{k_1+1}^{(\iota_{k_1}+1)}\right| < 4 \alpha_1(\lambda^*) + 4\alpha_2(\lambda^*)b.
	\end{equation}
	If $\lambda^*$ is sufficiently small, in particular, $\lambda^* \leq \min\left\{ \alpha_1^{-1}\left( \frac{\tau}{8} \right), \alpha_2^{-1}\left( \frac{\tau}{8b} \right) \right\}$, then we obtain that $\left|h_{k_1+1}^{(\iota_{k_1}-1)} - 2h_{k_1+1}^{(\iota_{k_1})} + h_{k_1+1}^{(\iota_{k_1}+1)}\right| < \tau$ from the bound in \eqref{eq:proof_conv_h_bound12}. This contradicts the necessary condition in \eqref{eq:proof_conv_h_condition} for $k_1$ to exist. By letting $\lambda^* = \min\left\{ \alpha_1^{-1}\left( \frac{\tau}{8} \right), \alpha_2^{-1}\left( \frac{\tau}{8b} \right),\alpha_2^{-1}\left(\frac{L_b}{4}\right), \lambda^{**} \right\}$, we have that $\lambda^* \leq \min\left\{ \alpha_1^{-1}\left( \frac{\tau}{8} \right), \alpha_2^{-1}\left( \frac{\tau}{8b} \right) \right\}$, which implies that there does not exist a $k_1 \in \mathbb{N}$ such that $u_{k_1+1} = u_{k_1}$. Thus, we always perturb. In addition, because $\lambda^* \leq \min\left\{\alpha_2^{-1}\left(\frac{L_b}{4}\right), \lambda^{**} \right\}$, it follows from the first part of Step 2 that all conditions of Lemma~\ref{lem:convergence_bounds} are met. This completes the proof of the theorem.   
\end{proof}


\addtolength{\textheight}{-12cm}   



\bibliography{library.bib}

\end{document}